\documentclass[conference]{IEEEtran}
\ifCLASSINFOpdf
\else
\fi
\usepackage{url}
\usepackage{mathrsfs,mathdots,mathtools,amsthm,amssymb,centernot,accents}
\usepackage{nicefrac}
\usepackage{stackengine} % for the \ubar command I have defined

\usepackage{hyperref}  %for references
\hypersetup{colorlinks,
            linkcolor=blue,
            citecolor=blue,
            urlcolor=black,
            anchorcolor=blue,
            filecolor=blue,
            plainpages=false,
            breaklinks=true}
\usepackage[nameinlink]{cleveref}

\usepackage{subcaption}

\usepackage[inline]{enumitem}
\setlist[enumerate]{
  leftmargin=*, 
  label=(\roman*)
}

\usepackage{bm}	%Command \bm{make bold}
\usepackage{multicol}
\usepackage[dvipsnames]{xcolor}
\usepackage{dirtytalk}
\usepackage{color, soul}
\usepackage{lipsum}
\usepackage{booktabs}

\usepackage[style=ieee,maxnames=2, minnames=1, maxbibnames=99,citestyle=numeric-comp, url=true]{biblatex}

\usepackage{algorithm}
\usepackage{algpseudocode}
\newcommand{\bE}{\ensuremath{\mathbb{E}}}

\newcommand{\bP}{\ensuremath{\mathbb{P}}}

\newcommand{\cA}{\ensuremath{\mathcal{A}}}
\newcommand{\cB}{\ensuremath{\mathcal{B}}}

\newcommand{\cE}{\ensuremath{\mathcal{E}}}

\newcommand{\cG}{\ensuremath{\mathcal{G}}}

\newcommand{\cJ}{\ensuremath{\mathcal{J}}}

\newcommand{\cL}{\ensuremath{\mathcal{L}}}

\newcommand{\cX}{\ensuremath{\mathcal{X}}}
\newcommand{\cY}{\ensuremath{\mathcal{Y}}}
\newcommand{\cZ}{\ensuremath{\mathcal{Z}}}

\newcommand{\ind}{\ensuremath{\mathbf{1}}}
\newcommand{\Ren}{R\'enyi }

\newcommand\ubar[1]{\stackunder[1.2pt]{$#1$}{\rule{.8ex}{.075ex}}}

\DeclarePairedDelimiter\abs{\lvert}{\rvert}%
\newtheorem{theorem}{Theorem}
\newtheorem{definition}{Definition}

\theoremstyle{remark}
\newtheorem{cor}{Corollary}

\newtheorem{lemma}{Lemma}

\newtheorem{example}{Example}

\allowdisplaybreaks

\definecolor{dark_green}{RGB}{6,64,43}

\begin{document}
\title{Information Leakage Envelopes}


\author{\IEEEauthorblockN{Sara Saeidian}
\IEEEauthorblockA{Inria, Palaiseau, France\\ KTH Royal Institute of Technology, Stockholm, Sweden\\ saeidian@kth.se}
\and
\IEEEauthorblockN{Carlos Pinzón}
\IEEEauthorblockA{Inria, École Polytechnique\\ Palaiseau, France\\ carlos.pinzon@inria.fr}
\and
\IEEEauthorblockN{Catuscia Palamidessi}
\IEEEauthorblockA{Inria, École Polytechnique\\ Palaiseau, France\\ catuscia.palamidessi@inria.fr}
\thanks{The work of Sara Saeidian was supported by the Swedish Research Council (VR) under grant 2024-06615. The work of Catuscia Palamidessi was supported by the ELSA project in the HORIZON EUROPE Framework Programme (project number 101070617).} 
}

\IEEEoverridecommandlockouts

%Telephone: (800) 555--1212, Fax: (888) 555--1212}

% \author{Anonymous Authors}

\maketitle

\begin{abstract}
We study privacy guarantees in the framework of \emph{pointwise maximal leakage} (PML) that satisfy two requirements: they are robust under post-processing and upper bound the failure probability, i.e., the probability that the information leakage exceeds a given threshold. We first examine two candidate definitions inspired by (approximate) differential privacy and show that neither one satisfies both requirements simultaneously. We then introduce the notion of the \emph{PML envelope}, which quantifies the largest amount of information leakage about a secret after arbitrary post-processing of a mechanism’s output. By construction, the PML envelope satisfies both requirements. We discuss basic structural properties of the envelope, such as monotonicity, and derive general upper and lower bounds. We further analyze the envelope for two widely used privacy mechanisms: the PML-extremal mechanisms in the high-privacy regime and randomized response. Overall, this work establishes the PML envelope as a natural and operationally meaningful definition for providing privacy guarantees that are preserved under arbitrary downstream transformations.
\end{abstract}


% \fi
\IEEEpeerreviewmaketitle

\section{Introduction}
The introduction of \emph{differential privacy} (DP) \cite{dworkCalibratingNoiseSensitivity} marked a major advance in privacy-preserving technologies. Rather than focusing on defenses against specific privacy attacks, such as re-identification or attribute disclosure, DP frames privacy as an inherent property of a data-processing system. Central to this perspective is the concept of \emph{privacy loss}~\cite{dworkAlgorithmicFoundationsDifferential2014}, defined as
\begin{equation*}
L_{x,x'}(Y) \coloneqq \log \frac{P_{Y \mid X = x}(Y)}{P_{Y \mid X = x'}(Y)},
\end{equation*}
where $x$ and $x'$ denote two possible values of the sensitive input (the \emph{secret}) and $P_{Y \mid X}$ is the mechanism that produces the released information $Y$. The privacy loss random variable captures how much evidence an observation $Y$ provides for distinguishing between $x$ and $x'$. 
% The goal of DP is then to allow making useful inferences about the input $X$ while controlling the privacy loss. 

The original formulation of DP, known as $\varepsilon$-DP, imposes a uniform upper bound on the privacy loss by a constant $\varepsilon > 0$ across all relevant pairs of inputs, that is, $L_{x,x'}(Y) \leq \varepsilon$~\cite{dworkCalibratingNoiseSensitivity}. This strong guarantee often requires adding noise with Laplace or geometric distributions. While these mechanisms provide effective privacy protection, Gaussian noise is often more appealing from a statistical standpoint. The faster tail decay of the Gaussian distribution typically translates into improved utility, and properties such as closedness under convolution simplify the theoretical analysis (since the sum of Gaussians remains Gaussian). These properties motivated the search for relaxations of $\varepsilon$-DP that allow the use of Gaussian noise.

A natural way to relax $\varepsilon$-DP is to allow the privacy loss to exceed $\varepsilon$ with a small probability. This idea led to the notion of \emph{probabilistic DP}~\cite{machanavajjhala2008privacy}, which requires 
\begin{equation*}
    \bP \{L_{x,x'}(Y) > \varepsilon \} \leq \delta
\end{equation*}
for some $\delta \in (0,1)$. Despite its intuitive appeal, this definition did not gain widespread use due to a fundamental limitation: probabilistic DP is not \emph{closed under post-processing}. That is, an adversary may apply a function to the output of a mechanism satisfying probabilistic DP (without access to the secret) to produce a new outcome that no longer satisfies the original guarantee. Closedness under post-processing is often treated as an axiom for privacy definitions and is motivated by the data-processing inequality in information theory~\cite{kifer2012axiomatic}. Intuitively, more processing cannot increase the information available about the secret, and therefore should preserve privacy guarantees. As a result, probabilistic DP was largely set aside in favor of \emph{approximate differential privacy} (ADP), also known as $(\varepsilon,\delta)$-DP. A mechanism satisfies $(\varepsilon, \delta)$-DP if
\begin{equation}
\label{eq:adp_intro}
    P_{Y \mid X=x}(\cA) \leq e^\varepsilon P_{Y \mid X=x'}(\cA) + \delta, 
\end{equation}
for all (measurable) sets $\cA$. This definition is closed under post-processing, accommodates the use of Gaussian noise, and also enables \emph{advanced composition}~\cite{dworkBoostingDifferentialPrivacy2010}. 

The parameter $\delta \in (0,1)$ in approximate differential privacy is often informally interpreted as restricting the \emph{failure probability} at level $\varepsilon$, i.e., $\mathbb{P}\{ L_{x,x'}(Y) > \varepsilon \}$. This interpretation is, however, somewhat loose. One can show that $\delta$ merely provides a \emph{lower bound} on this probability, and in fact, \cite{meiser2018approximate} constructed an example in which $\mathbb{P}\{ L_{x,x'}(Y) > \varepsilon \}$ was much larger than the $\delta$ in ADP. Instead, the parameter $\delta$ in ADP upper bounds the failure probability at
larger thresholds, such as $2\varepsilon$ or, more generally, $k\varepsilon$ (see~\eqref{eq:adp_tail_k} for the precise statement). Note that while a constant-factor increase in $\delta$ is generally harmless (since $\delta$ is chosen to be cryptographically small), a constant-factor increase in $\varepsilon$ exponentially weakens the privacy guarantee as \eqref{eq:adp_intro} depends on $e^\varepsilon$. 

% That said, this difficulty in interpreting $\delta$ as a failure probability has not posed major problems in practice. In most modern systems, privacy loss is not directly tracked using ADP. Instead, one typically relies on accounting frameworks based on Rényi DP~\cite{mironovRenyiDifferentialPrivacy2017a}, concentrated DP~\cite{bunConcentratedDifferentialPrivacy2016a}, or Gaussian DP~\cite{dongGaussianDifferentialPrivacy2022}, which control the distribution of the privacy loss random variable. The resulting
% guarantees are often converted to $(\varepsilon,\delta)$-DP, but this is primarily done for reporting and comparison purposes.

Interestingly, the difficulties associated with defining relaxed privacy guarantees based on a probability of failure are not unique to DP, but have been observed in a different framework based on a privacy measure called \emph{pointwise maximal leakage} (PML) \cite{saeidian2023pointwise_it}. PML is a measure with several attractive properties. Privacy guarantees based on PML allow the release of population-level features of the data while protecting its nuanced and instance-dependent features~\cite{inferential_privacy}. Moreover, connections between PML-based privacy and DP have been established: $\varepsilon$-DP is equivalent to bounding the PML of every record in a database under all product distributions, i.e., distributions under which records are independent~\cite{inferential_privacy}. 

The PML framework is built around a central random quantity, the \emph{information leakage random variable}, which can be expressed as 
\begin{equation*}
    \ell(X \to Y) = \log \, \max_x \, \frac{P_{X \mid Y}(x \mid Y)}{P_X(x)},
\end{equation*}
where $P_{X \mid Y}$ is the posterior distribution of $X$ given $Y$. Different notions of privacy correspond to placing different restrictions on this quantity. In particular, the guarantee known as $\varepsilon$-PML imposes a uniform upper bound on the information leakage, in direct analogy with $\varepsilon$-DP.

In \cite{saeidian2023pointwise_it}, the authors proposed a relaxation of $\varepsilon$-PML based on bounding the failure probability at level $\varepsilon$, that is 
\begin{equation*}
    \bP \{ \ell (X \to Y) > \varepsilon \}.
\end{equation*}
However, it turns out that much like probabilistic DP, this relaxation is not closed under post-processing. To address this issue, the authors introduced an alternative privacy guarantee based on bounding the information leaked to \emph{events} (i.e., subsets of the output space). This notion, however, is not directly comparable to the failure probability: depending on the setup, either quantity may dominate the other \cite{saeidian2023pointwise_it}.

In this paper, we undertake the task of defining a relaxation of $\varepsilon$-PML that satisfies two desirable properties:
\begin{enumerate}[leftmargin=*]
    \item the slack parameter $\delta$ should provide an \emph{upper bound} on the failure probability at level $\varepsilon$, and
    \item the resulting privacy guarantee should be closed under post-processing. 
\end{enumerate}
As the preceding discussion illustrates, existing relaxations typically satisfy one of these properties, and achieving both properties simultaneously appears to be non-trivial. 

\subsection{Contributions}
We first examine two relaxations of $\varepsilon$-PML inspired by ADP. Although ADP itself does not satisfy property~(i), this investigation is nevertheless informative, as it reveals the types of PML-based definitions that arise when one follows the ADP design philosophy. Of the two candidate definitions, the first again fails to be post-processing safe (Theorem~\ref{thm:dp_like_deltas}). The second is post-processing safe, but does not admit a clear or consistent relationship with the failure probability. Consequently, neither definition satisfies both properties~(i) and~(ii).

Our main contribution is introducing the concept of the \emph{PML envelope}. The key idea is to evaluate the failure probability after all possible downstream transformations, instead of just at the original output $Y$; in other words, to \say{close} the failure probability. First, we define  
\begin{equation*}
    \delta_c(\varepsilon) \coloneqq \sup_{Z} \; \bP \{\ell(X \to Z) > \varepsilon\}, \quad \varepsilon >0,
\end{equation*}
where the supremum is over all deterministic and randomized functions $Z$ of $Y$. This quantity represents the largest probability that the information leakage exceeds $\varepsilon$ after arbitrary post-processing. The PML envelope is then defined as
\begin{equation*}
    \varepsilon_c(\delta) \coloneqq \inf\{\varepsilon > 0 : \delta_c(\varepsilon) \le \delta\}, \quad \delta \in (0,1),
\end{equation*}
that is, the smallest leakage threshold $\varepsilon$ for which the closed failure probability is at most $\delta$. By definition, the PML envelope satisfies both properties (i) and (ii) above. We provide equivalent characterizations of $\varepsilon_c(\delta)$ in terms of the cumulative distribution function (CDF) of the information leakage random variable $\ell(X \to Y)$ (Theorem~\ref{thm:comp_eps}). We also establish basic properties of $\varepsilon_c$, including monotonicity and lower semi-continuity (Lemma~\ref{lemma:properties_envelope}). 

In general, computing the PML envelope exactly may be difficult. For this reason, we derive general upper and lower bounds on $\varepsilon_c$. The upper bound (Theorem~\ref{thm:upper_bound_envelope}) is expressed in terms of the \emph{multiplicative Bayes capacity}~\cite{alvim2014additive} (also known as \emph{maximal leakage}~\cite{issaOperationalApproachInformation2020}), which is a well-studied quantity in the quantitative information flow literature~\cite{alvim2020science}. Lower bounds are obtained by restricting the class of admissible post-processings. In particular, we derive a lower bound based on \emph{binary} post-processings, which is simple and efficiently computable. This lower bound coincides with the event-based privacy guarantee introduced in~\cite{saeidian2023pointwise_it}; we discuss this connection in detail in Appendix~\ref{sec:events}.

Finally, we illustrate how the PML envelope can be computed and bounded by analyzing two canonical mechanisms. 
The first is the class of \emph{PML-extremal} mechanisms in the high-privacy regime~\cite{grosseExtremalMechanismsPointwise2024}. These mechanisms are the utility-optimal solutions to a broad class of optimization problems under the $\varepsilon$-PML constraint with sufficiently small $\varepsilon >0$. For these mechanisms, we characterize the envelope exactly (Theorem~\ref{thm:pml_ext}) and show that $\varepsilon_c(\delta) = \varepsilon$ for all $\delta \in (0,1)$. 

The second is the \emph{randomized response} mechanism~\cite{warnerRandomizedResponseSurvey1965}, a standard and widely used tool for guaranteeing local differential privacy~\cite{kairouzExtremalMechanismsLocal2016}. Randomized response is a natural benchmark due to its simplicity, in particular, its symmetric structure, as well as its widespread use. For this mechanism, we derive upper and lower bounds on the PML envelope (Theorem~\ref{thm:k-rr}). Comparing these bounds to the corresponding ADP guarantees of randomized response shows that the two frameworks can behave quite differently, with no consistent relationship in general. This is not surprising because in the PML envelope, $\delta$ represents the failure probability, whereas in ADP it appears as an additive slack parameter.

\section{Background}
\label{sec:background}
\subsection{Notation}
Uppercase letters denote random variables, lowercase letters denote their realizations, and calligraphic letters denote sets. All sets are assumed to be finite. We use $X$ to denote a random variable containing sensitive information, also referred to as the \emph{secret}. Its probability distribution is denoted by $P_X$, and its domain by $\mathcal X$. A privacy mechanism (or simply, a mechanism) is specified by a conditional probability distribution $P_{Y \mid X}$, which takes $X$ as input and produces an output $Y$ with domain $\mathcal Y$ and (marginal) distribution $P_Y$. The joint distribution of $(X,Y)$ is denoted by $P_{XY}$. 

Random variables $X$, $Y$, and $Z$ are said to form a Markov chain $X - Y - Z$ if $X$ and $Z$ are conditionally independent given $Y$, that is, $P_{XZ \mid Y} = P_{X \mid Y} \times P_{Z \mid Y}$. We write $P_{Z \mid X} = P_{Z \mid Y} \circ P_{Y \mid X}$ to denote \emph{marginalization}, meaning that
\begin{equation*}
P_{Z \mid X=x}(z) = \sum_{y \in \mathcal Y} P_{Z \mid Y=y}(z)\, P_{Y \mid X=x}(y), \quad x \in \cX, z \in \cZ. 
\end{equation*}
For a set $\mathcal E$, $\ind_{\mathcal E}$ denotes its indicator function. For a positive integer $k$, we write $[k] = \{1,\dots,k\}$. 

%%%%%%%%%%%%%%%
%%%%%%%%%%%%%%%
\subsection{Differential Privacy}
% Differential privacy (DP) was introduced as a framework for releasing statistics about databases while protecting the privacy of individual entries~\cite{dworkCalibratingNoiseSensitivity,dworkAlgorithmicFoundationsDifferential2014}. 
The original definition of DP states that an adversary should not be able to distinguish between \emph{neighboring} datasets, i.e., databases that differ in a single record, based on the statistics released from each~\cite{dworkCalibratingNoiseSensitivity,dworkAlgorithmicFoundationsDifferential2014}. Later, DP was extended to decentralized settings, where data is perturbed before collection. This variant is known as \emph{local differential privacy} (LDP)~\cite{kasiviswanathan2011can,duchi2013LDPminmaxDEF}. In this work, we adopt the local model for all DP definitions. This choice is made without loss of generality, since switching between the two models amounts to redefining what it means for two inputs to be \say{neighbors.} By working in the local model, we abstract away the structure of the secret and instead focus on our main objective: to analyze how various definitions behave under post-processing. 

Given $x, x' \in \cX$, let 
\begin{equation*}
    L_{x,x'} (Y) \coloneqq \log \frac{P_{Y \mid X=x}(Y)}{P_{Y \mid X=x'}(Y)},
\end{equation*}
denote the \emph{privacy loss random variable} of DP~\cite{dworkAlgorithmicFoundationsDifferential2014}. The simplest DP definition imposes a uniform bound on this loss. 

\begin{definition}[Pure DP~\cite{dworkCalibratingNoiseSensitivity}]
Let $\varepsilon > 0$. A privacy mechanism $P_{Y \mid X}$ is said to satisfy $\varepsilon$-DP if $L_{x,x'} (y) \leq \varepsilon$ for all $x,x' \in \cX$ and all $y \in \cY$.   
\end{definition}
A standard relaxation of pure DP, known as approximate DP (ADP)~\cite{dworkOurDataOurselves2006a}, allows an additive slack $\delta \in (0,1)$ in the event-wise comparison of the conditional distributions. 
\begin{definition}[Approximate DP~\cite{dworkAlgorithmicFoundationsDifferential2014}]
\label{def:approx_dp}
Let $\varepsilon > 0$ and $\delta \in (0,1)$. A privacy mechanism $P_{Y \mid X}$ is said to satisfy $(\varepsilon, \delta)$-DP if for all $x,x' \in \cX$ and all sets $\cE \subset \cY$ we have 
\begin{equation}
\label{eq:approx_ldp}
    P_{Y \mid X=x}(\cE) \leq e^\varepsilon P_{Y \mid X=x'}(\cE) + \delta. 
\end{equation}
\end{definition}
The common interpretation of $(\varepsilon, \delta)$-DP is that it allows the privacy guarantees of $\varepsilon$-DP to fail with probability $\delta$. However, as pointed out in~\cite{meiser2018approximate}, this interpretation is misleading: the parameter $\delta$ is, in fact, a \emph{lower bound} on the worst-case failure probability. To demonstrate this, we find the smallest $\delta$ that satisfies \eqref{eq:approx_ldp} for a fixed $\varepsilon$. This function, called the \emph{privacy profile}~\cite{ballePrivacyAmplificationSubsampling2018} and denoted by $\delta^*(\varepsilon, P_{Y \mid X})$, can be expressed as
\begin{equation*}
    \delta^*(\varepsilon, P_{Y \mid X}) = \max_{x,x' \in \cX} \, \max_{\cE \subset \cY} \; \Big( P_{Y \mid X=x}(\cE) - e^\varepsilon P_{Y \mid X=x'}(\cE) \Big). 
\end{equation*}
Then, $P_{Y \mid X}$ satisfies $(\varepsilon, \delta)$-DP for all $\delta^*(\varepsilon) \leq \delta < 1$. It is not difficult to see that $\delta^*$ can also be expressed in the following more intuitive form~\cite[Lemma~2.8]{canonneDiscreteGaussianDifferential2022}: 
\begin{align}
\label{eq:delta_star_dp}
    &\delta^*(\varepsilon, P_{Y \mid X})\nonumber\\
    &= \max_{x,x'} \bE_{Y \sim P_{Y \mid X=x}} \left [\max \left\{0, 1 - \frac{\exp(\varepsilon)}{\exp(L_{x,x'} (Y))} \right\} \right]\\[0.5em]
    &= \max_{x,x'} \bE_{Y \sim P_{Y \mid X=x}} \!\left [\ind_{\left\{L_{x,x'} > \varepsilon \right\}}(Y) \!\left(1 - \frac{\exp(\varepsilon)}{\exp(L_{x,x'} (Y))} \right) \!\right].\nonumber 
\end{align}
Intuitively, the term
\begin{equation*}
    1 - \frac{\exp(\varepsilon)}{\exp(L_{x,x'}(y))},
\end{equation*}
is a penalty applied to outcomes $y$ with privacy loss larger than $\varepsilon$. The greater the deviation of privacy loss from $\varepsilon$, the larger this penalty becomes, reflecting that such outcomes contribute more significantly to the overall privacy risk.
Observe that $0 \leq 1 - \frac{\exp(\varepsilon)}{\exp(L_{x,x'} (y))} \leq 1$ on the set $\left\{L_{x,x'} > \varepsilon \right\}$, so we have
\begin{align*}
    \delta^*(\varepsilon, P_{Y \mid X}) &\leq \max_{x,x' \in \cX} \; \bE_{Y \sim P_{Y \mid X=x}} \left [\ind_{\left\{L_{x,x'} > \varepsilon \right\}}(Y) \right]\\[0.5em]
    &= \max_{x,x' \in \cX} \; P_{Y \mid X=x} \left\{L_{x,x'}(Y) > \varepsilon \right\}.
\end{align*}
The quantity $\max_{x,x'} \; P_{Y \mid X=x} \left\{L_{x,x'}(Y) > \varepsilon \right\}$ represents the worst-case failure probability of DP, that is, the (maximum) probability that the privacy loss exceeds $\varepsilon$. Hence, the above derivation shows that the value of $\delta$ in ADP merely lower bounds this quantity. It is nevertheless true that $(\varepsilon, \delta)$-DP implies the weaker tail bound 
\begin{equation*}
    P_{Y \mid X=x} \big\{L_{x,x'}(Y) > 2 \varepsilon \big\} \leq \frac{\delta}{1 - e^{-\varepsilon}},
\end{equation*}
for all $x,x'$~\cite[Lemma 3.3]{kasiviswanathanSemanticsDifferentialPrivacy2014} (see also \cite{jonanyCorrectingMisconceptionDifferential2022}). More generally, $(\varepsilon, \delta)$-DP implies
\begin{equation}
\label{eq:adp_tail_k}
    \max_{x,x'} \; P_{Y \mid X=x} \big\{L_{x,x'}(Y) > k \varepsilon \big\} \leq \frac{\delta}{1 - e^{-(k-1)\varepsilon}},
\end{equation}
for all integers $k \geq 2$.  
While some recent works have taken care to be more precise in their description of ADP, for instance, by stating that ADP restricts the failure probability up to a scaling of the parameters, the misconception remains widespread, especially in more applied contexts where the focus is not on the theoretical nuances.

The above discussion naturally raises the question of whether we could define a DP variant by explicitly upper bounding the failure probability. Such a definition was, in fact, proposed by~\textcite{machanavajjhala2008privacy} and is known as \emph{probabilistic DP}.
\begin{definition}[Probabilistic DP~\cite{machanavajjhala2008privacy}]
\label{def:probabilitist_dp}
Let $\varepsilon > 0$ and $\delta \in (0,1)$. A privacy mechanism $P_{Y \mid X}$ is said to satisfy $(\varepsilon, \delta)$-probabilistic DP if for all $x, x' \in \cX$ we have $P_{Y \mid X=x} \left\{L_{x,x'}(Y) > \varepsilon \right\} \leq \delta$. 
\end{definition}

While probabilistic DP is arguably more intuitive than ADP, it is rarely used in practice. This is primarily because probabilistic DP may not be preserved after post-processing~\cite{kiferAxiomaticViewStatistical2012,meiser2018approximate}, in contrast to ADP, which is post-processing safe~\cite{dworkAlgorithmicFoundationsDifferential2014}. This property has contributed to ADP becoming the de facto standard in the literature. 

In the following section, we discuss similar challenges in the framework of pointwise maximal leakage.

\subsection{Pointwise Maximal Leakage} 
PML~\cite{saeidian2023pointwise_it} is a recent notion of privacy defined using concepts from \emph{quantitative information flow}~\cite{alvim2020science}. PML quantifies the inference risk posed by a broad class of adversaries. Its threat model can be described as follows: consider an adversary who seeks to maximize a non-negative gain function $g$ by producing a guess $W$ of the private variable $X$. The gain function encodes the adversary's objective and can capture a wide range of privacy attacks, including membership and attribute inference~\cite{saeidian2023pointwise_it}. For an output $y \in \mathcal Y$, PML measures information leakage as the ratio between the adversary’s expected gain after observing $y$ and the expected gain before observing $y$. Then, to obtain a robust and attack-agnostic notion of leakage, this ratio is maximized over all nonnegative gain functions. Formalizing this idea leads to the following definition for PML.
\begin{definition}[PML~\cite{saeidian2023pointwise_it}]
\label{def:PML}
Suppose $X \sim P_X$ and let $Y$ be the random variable induced by the mechanism $P_{Y \mid X}$. The pointwise maximal leakage from $X$ to $y \in \mathcal{Y}$ is defined as
\begin{equation}
\label{eq:g-leakage}
    \ell_{P_{XY}}(X \to y) := \log \; \sup_{g} \; \frac{\sup\limits_{P_{W \mid Y}} \mathbb{E}[g(X, W) \mid Y = y]}{\max_{w' \in \mathcal{W}} \mathbb{E}[g(X, w')]},
\end{equation}
where $P_{W \mid Y}$ is the conditional distribution of the adversary's guess $W$ given $Y$. The supremum is over all non-negative measurable functions $g$.
\end{definition}
In this work, both $X$ and $Y$ are assumed to be finite-valued. Under this assumption, it was shown in~\cite{saeidian2023pointwise_it} that PML takes the simpler form
\begin{align*}
    \ell_{P_{XY}}\!(X \!\to \!y) = \log \max_{x \in \mathcal X} \frac{P_{X \mid Y=y}(x)}{P_X(x)} = \log \max_{x \in \mathcal X} \frac{P_{Y \mid X=x}(y)}{P_Y(y)},
\end{align*}
where $P_{X \mid Y}$ denotes the posterior distribution of $X$ given $Y$. It is straightforward to see that PML satisfies the bounds 
\begin{equation}
\label{eq:pml_bounds}
    0 \leq \ell_{P_{XY}}(X \to y) \leq \log \frac{1}{\min_{x \in \cX} P_X(x)},
\end{equation}
for all $y \in \cY$. 

In the information theory literature, the quantity 
\begin{equation*}
    i_{P_{XY}}(x;y) = \log \frac{P_{XY}(x,y)}{P_X(x) P_Y(y)}, \quad x \in \cX, y \in \cY,
\end{equation*}
is commonly referred to as the \emph{information density} of $P_{XY}$. PML can also be expressed as 
\begin{equation*}
    \ell_{P_{XY}}(X \to y) = \max_{x \in \cX} \; i_{P_{XY}}(x;y).
\end{equation*}
Note that, unlike DP, PML depends on the prior distribution $P_X$ and is therefore a property of the joint distribution $P_{XY}$. When the joint distribution is clear from context, we omit the subscript and write $i(x;y)$ for information density and $\ell(X \to y)$ or simply $\ell(y)$ for PML.

The joint distribution $P_{XY}$ is said to satisfy $\varepsilon$-PML with $\varepsilon>0$ if $\ell(X \to y) \leq \varepsilon$ for all $y \in \mathcal Y$. In~\cite{saeidian2023pointwise_it}, the authors also introduced a relaxation of $\varepsilon$-PML by imposing an upper bound on the tail of $\ell(X \to Y)$.\footnote{The function $\ell(X \to y)$ is defined pointwise for each $y \in \mathcal Y$. Consequently, $\ell(X \to Y)$ is a random variable induced by $Y$.} 

\begin{definition}[Probabilistic PML]
Let $\varepsilon > 0$ and $\delta \in (0,1)$. The joint distribution $P_{XY}$ is said to satisfy $(\varepsilon, \delta)$-probabilistic PML if 
\begin{equation*}
    P_Y \{\ell(X \to Y) > \varepsilon\} \leq \delta. 
\end{equation*}
\end{definition}
Much like probabilistic LDP, probabilistic PML is \emph{not} closed under post-processing. To illustrate this, below we give an example similar to~\cite[Example 7]{saeidian2023pointwise_it}.

\begin{example}
\label{ex:prob_pml_not_post_proc}
Let $X$ be uniformly distributed on $\cX = [4]$. Consider the privacy mechanism 
\begin{equation}
\label{eq:ex_mech}
    P_{Y \mid X} = \begin{bmatrix}
    0 & 0 & 0.5 & 0.5\\[0.5em]
    0 & 0 & 0.5 & 0.5\\[0.5em]
    0 & 0.2 & 0.4 & 0.4\\[0.5em]
    0.2 & 0 & 0.4 & 0.4
    \end{bmatrix}, 
\end{equation}
where $(P_{Y \mid X})_{ij} = P_{Y \mid X=i}(j)$. The outcomes have information leakage 
\begin{gather*}
    \ell_{P_{XY}}(X \to 1) = \ell_{P_{XY}}(X \to 2) = \log 4, \\
    \ell_{P_{XY}}(X \to 3) = \ell_{P_{XY}}(X \to 4) = \log \frac{10}{9}.
\end{gather*}
Since $P_Y(1) = P_Y(2) = \frac{1}{20}$, $P_{XY}$ satisfies $(\log \frac{10}{9}, 0.1)$-probabilistic PML. Now, let $Z = h(Y)$, where 
\begin{equation*}
    h(y) = 
    \begin{cases}
        1 & \text{if } y \in \{1,3\}, \\
        2 & \text{if } y \in \{2,4\},
    \end{cases} 
\end{equation*}
The outcomes of $Z$ are equiprobable and have information leakage $\ell_{P_{XZ}}(X \to 1) = \ell_{P_{XZ}}(X \to 2) = \log \frac{6}{5}$. Since $\frac{6}{5} > \frac{10}{9}$, $P_{XZ} = P_{Z \mid Y} \circ P_{XY}$ does \emph{not} satisfy the original guarantee of $(\log \frac{10}{9}, 0.1)$-probabilistic PML. 
\end{example}
In light of Example~\ref{ex:prob_pml_not_post_proc}, we ask: What alternative definitions can reconcile PML-based privacy guarantees with the post-processing requirement? We explore answers to this question in the subsequent sections.

\section{Two Candidate Definitions of\\ Approximate PML}
\label{sec:candidates}
As a first step, we examine two candidate definitions of approximate PML, inspired by ADP. In particular, we define analogues of~\eqref{eq:delta_star_dp} by replacing the DP privacy loss with either the information leakage random variable or the information density. 

Given $\varepsilon > 0$, define
\begin{align*}
   \psi_1(\varepsilon, P_{XY}) &\coloneqq \bE_{Y \sim P_Y} \left [\max \left\{0, 1 - \frac{\exp(\varepsilon)}{\exp(\ell(Y))} \right\} \right]\\[0.5em]
    &= \bE_{Y \sim P_Y} \left [\ind_{\left\{\ell > \varepsilon \right\}}(Y) \left(1 - \frac{\exp(\varepsilon)}{\exp(\ell(Y))} \right) \right],
\end{align*}
and 
\begin{align*}
    &\psi_2(\varepsilon, P_{XY}) \\
    &\coloneqq \max_{x} \bE_{Y \sim P_{Y \mid X=x}} \left [\max \left\{0, 1 - \frac{\exp(\varepsilon)}{\exp(i(x;Y))} \right\} \right]\\[0.5em]
    &= \max_{x} \bE_{Y \sim P_{Y \mid X=x}} \left [\ind_{\left\{ i(x;\cdot) > \varepsilon \right\}}(Y) \left(1 - \frac{\exp(\varepsilon)}{\exp(i(x;Y))} \right) \right].
\end{align*}
Both definitions follow the same general pattern: they assign a penalty to outcomes where the information leakage $\ell(y)$ or the information density $i(x;y)$ exceeds the threshold $\varepsilon$. Note that $0 \leq 1 - \frac{\exp(\varepsilon)}{\exp(\ell(y))} \leq 1$ when $\ell(y) > \varepsilon$, so $\psi_1$ lower bounds the PML failure probability $P_Y \{\ell(Y) > \varepsilon\}$.

Despite the parallelism in their expressions, these two definitions behave differently under post-processing: $\psi_1$ is not post-processing safe, whereas $\psi_2$ is.
\begin{theorem}
\label{thm:dp_like_deltas}
Let $\varepsilon > 0$.
\begin{enumerate}
\item There exist random variables $X, Y, Z$ satisfying the Markov chain $X-Y-Z$ such that 
\begin{equation*}
   \psi_1(\varepsilon, P_{Z \mid Y} \circ P_{XY}) >\psi_1(\varepsilon, P_{XY}).
\end{equation*}
\item  For all random variables $X, Y, Z$ satisfying the Markov chain $X-Y-Z$ we have
\begin{equation*}
   \psi_2(\varepsilon, P_{Z \mid Y} \circ P_{XY}) \leq\psi_2(\varepsilon, P_{XY}).
\end{equation*}
\end{enumerate}
\end{theorem}

\begin{IEEEproof}
\begin{enumerate}
\item It suffices to construct an example where $\psi_1(\varepsilon, P_{Z \mid Y} \circ P_{XY}) > \psi_1(\varepsilon, P_{XY})$. Recall the setup of Example~\ref{ex:prob_pml_not_post_proc}. Setting $\varepsilon = \log \frac{10}{9}$, a direct calculation yields
\begin{equation*}
    \psi_1(\varepsilon, P_{XY}) = \frac{13}{180} < \frac{2}{27} = \psi_1(\varepsilon, P_{XZ}).
\end{equation*}
Hence, $\psi_1$ can increase under post-processing. 

%%%%%%%%%%%%%%%%%%%%%%%%
\item Our argument mirrors the standard proof of post-processing for ADP. Consider the Markov chain $X-Y-Z$. Fix an arbitrary $x \in \cX$ and observe that
\begin{align*}
    &\bE_{Z \sim P_{Z \mid X=x}} \left [\ind_{\left\{ i(x;\cdot) > \varepsilon \right\}}(Z) \left(1 - \frac{\exp(\varepsilon)}{\exp(i(x;Z))} \right) \right]\\[0.3em]  
    &= \sum_{z : i(x;z) > \varepsilon} \left(1 - \frac{\exp(\varepsilon)}{\exp(i(x;z))} \right) \, P_{Z \mid X=x}(z)\\[0.3em]  
    &= \sum_{z : i(x;z) > \varepsilon} P_{Z\mid X=x}(z) - e^\varepsilon \sum_{z: i(x;z) > \varepsilon} P_Z(z)\\[0.3em]  
    &= P_{Z\mid X=x}\{z : i(x;z) > \varepsilon\} - e^\varepsilon P_Z\{z : i(x;z) > \varepsilon \}. 
\end{align*}
Thus, to prove that $\psi_2$ does not increase under post-processing, it suffices to show that 
\begin{equation*}
    P_{Z \mid X=x}(\cA) - e^\varepsilon P_Z(\cA) \leq \psi_2(\varepsilon, P_{XY}), 
\end{equation*}
for all $x \in \cX$ and arbitrary sets $\cA \subseteq \cZ$. Indeed, for each set $\mathcal A \subseteq \cZ$ we have 
\begin{align*}
    &P_{Z \mid X=x}(\cA) - e^\varepsilon P_Z(\cA)\\[0.3em]  
    &= \sum_{y \in \cY} P_{Z \mid Y=y}(\cA) \Big(P_{Y \mid X=x}(y) - e^\varepsilon P_Y(y) \Big) \\[0.3em]
    &\leq \sum_{y : \, i(x;y) > \varepsilon} P_{Z \mid Y=y}(\cA)  \Big(P_{Y \mid X=x}(y) - e^\varepsilon P_Y(y) \Big)\\[0.3em] 
    &\leq \sum_{y : \, i(x;y) > \varepsilon}  P_{Y \mid X=x}(y) - e^\varepsilon P_Y(y)\\[0.3em]  
    &= \sum_{y : \, i(x;y) > \varepsilon} \left(1 - \frac{\exp(\varepsilon)}{\exp(i(x;y))} \right) \, P_{Y \mid X=x}(y)\\[0.3em]  
    &= \bE_{Y \sim P_{Y \mid X=x}} \left [\ind_{\left\{ i(x;\cdot) > \varepsilon \right\}}(Y) \left(1 - \frac{\exp(\varepsilon)}{\exp(i(x;Y))} \right) \right]\\[0.3em]  
    &\leq \psi_2(\varepsilon, P_{XY}). 
\end{align*}
\end{enumerate}   
\end{IEEEproof}

Thus, Theorem~\ref{thm:dp_like_deltas} establishes that $\varepsilon$-PML can be relaxed with an additive parameter such that the resulting definition is closed under post-processing. It is straightforward to see that $\psi_2$ can be written as 
\begin{equation*}
    \psi_2(\varepsilon, P_{XY}) = \max_{x \in \cX} \, \max_{\cE \subset \cY} \; \Big( P_{Y \mid X=x}(\cE) - e^\varepsilon P_Y(\cE) \Big). 
\end{equation*}

Although $\psi_2$ is closed under post-processing, it does not provide a suitable proxy for the failure probability $P_Y\{\ell(Y) > \varepsilon\}$, since neither quantity bounds the other one in general. To illustrate this, consider again the mechanism $P_{Y \mid X}$ in~\eqref{eq:ex_mech}. We have
\begin{equation*}
   P_Y\{\ell(Y) > \log \tfrac{10}{9}\} = P_Y\{\ell(Y) > \log 3\} = 0.1,
\end{equation*}
while at the same time,
\begin{equation*}
    \psi_2(\log 3, P_{XY}) = 0.05
    \;<\;
    0.1
    \;<\;
    \frac{13}{90}
    = \psi_2\!\left(\log \tfrac{10}{9}, P_{XY}\right).
\end{equation*}

In the next sections, we take a different approach and focus on directly \say{closing} the tail probability $P_Y\{\ell(Y) > \varepsilon\}$.

\section{Closing the Probability of Failure and the PML Envelope}
\label{sec:closing}
We now turn to a natural approach to the post-processing question, one that directly addresses the core issue of probabilistic PML. Instead of defining an additive relaxation or adjusting penalties, we examine the worst-case probability of failure across all possible post-processing mechanisms.

Formally, given a joint distribution $P_{XY}$ and $\varepsilon > 0$, define 
\begin{equation}
\label{eq:closed_delta}
    \delta_c(\varepsilon) \coloneqq \sup_{Z:X-Y-Z} \bP \{\ell(Z) > \varepsilon\},
\end{equation}
where the supremum is taken over all finite random variables $Z$ satisfying the Markov chain $X - Y - Z$, alternatively, all conditional distributions $P_{Z \mid Y}$. Observe that $\delta_c$ quantifies the largest probability that PML exceeds $\varepsilon$ under any downstream transformation, so by definition, it is post-processing safe. 

In \eqref{eq:closed_delta}, we fix $\varepsilon$ and find the largest failure probability. Alternatively, we could fix $\delta \in (0,1)$ and find the \emph{smallest} $\varepsilon$ that holds with probability at least $1 - \delta$ after arbitrary post-processing. Let $Z$ denote a (possibly randomized) function of $Y$, and consider its corresponding leakage random variable $\ell(Z)$. Let 
\begin{equation*}
    C_Z(t) = \bP \{\ell(Z) \leq t\}, \quad t \geq 0,
\end{equation*}
denote the \emph{cumulative distribution function} (CDF) of $\ell(Z)$, and for $s \in (0,1)$, define
\begin{equation*}
    \operatorname{Quant}^{\leftarrow}_Z(s) \coloneqq \inf \bigl \{t \geq 0 : C_Z(t) \geq s \bigr\}.
    \end{equation*}
to be the \emph{left-continuous quantile function} of $\ell(Z)$ at level $s$. Then, for $\delta \in (0,1)$, we define
\begin{equation}
\label{eq:lcq_func}
  \ubar \varepsilon_Z(\delta) \coloneqq \operatorname{Quant}^{\leftarrow}_Z(1-\delta) = \inf \bigl \{t \geq 0 : C_Z(t) \geq 1-\delta \bigr\},
\end{equation}
which captures the smallest threshold $t$ such that PML is bounded by $t$ with probability at least $1 - \delta$. Observe that the mapping $s \mapsto \operatorname{Quant}^{\leftarrow}_Z(s)$ is non-decreasing and left-continuous. Consequently, $\delta \mapsto \ubar \varepsilon_Z(\delta)$ is non-increasing and right-continuous. The function $\ubar \varepsilon_Z(\delta)$ also admits the equivalent, and somewhat more explicit formulation
\begin{equation}
\label{eq:var_lcq_func}
    \ubar \varepsilon_Z(\delta) = \min_{\substack{\mathcal{A} \subset \cZ \\ P_Z(\mathcal{A}) \geq 1 - \delta}} \; \max_{z \in \mathcal{A}} \, \ell(z), \quad \delta \in (0,1).
\end{equation}
That is, $\ubar \varepsilon_Z(\delta)$ tells us how small we can make the worst-case leakage, up to ignoring a set of probability $\delta$. The equivalence between~\eqref{eq:lcq_func} and \eqref{eq:var_lcq_func} is proved in Lemma~\ref{lemma:lcq_formulations} in Appendix~\ref{ssec:quantile_equivalence_proof} for completeness. Note that $P_{XY}$ satisfies $(\varepsilon, \delta)$-probabilistic PML if and only if $\ubar \varepsilon_Y(\delta) \leq \varepsilon$. 

We now define the \emph{PML envelope} of $P_{XY}$, denoted by $\varepsilon_c$, as the supremum of $\ubar \varepsilon_Z$ over all post-processings of $Y$, i.e., 
\begin{equation*}
    \varepsilon_c(\delta) \coloneqq \sup_{Z : X - Y - Z} \ubar \varepsilon_Z(\delta), \quad \delta \in (0,1).
\end{equation*}
In words, $\varepsilon_c(\delta)$ captures the tightest privacy guarantee that survives arbitrary downstream transformations, up to a failure probability of $\delta$. 

By definition, the map $\delta \mapsto \varepsilon_c(\delta)$ is non-increasing. It is also easy to show that $\varepsilon_c(\delta)$ is lower semi-continuous: $\ubar \varepsilon_Z$ is lower semi-continuous because it is non-increasing and right-continuous, and the supremum of an arbitrary family of lower semi-continuous functions is itself lower semi-continuous. We therefore obtain the following result.
\begin{lemma}
\label{lemma:properties_envelope}
The PML envelope $\varepsilon_c$ is non-increasing and lower semi-continuous on $(0,1)$.       
\end{lemma}

Before we characterize and compute the PML envelope, let us first express it in an alternative form. Given a random variable $Z$ and $s \in (0,1)$, let
\begin{equation*}
    \operatorname{Quant}^{\rightarrow}_Z(s) \coloneqq \sup \bigl \{t \geq 0 : C_Z(t) \leq s \bigr\},
    \end{equation*}
be the \emph{right-continuous quantile function}\footnote{See \cite{embrechts2013note,dufour1995distribution} for a review of generalized inverses and quantile functions.} of $\ell(Z)$ at level $s$. Then, for $\delta \in (0,1)$ define     
\begin{align*}
   \bar \varepsilon_Z(\delta) &\coloneqq \operatorname{Quant}^{\rightarrow}_Z(1-\delta) = \sup \{t \geq 0 : C_Z(t) \leq 1 - \delta \} \\[0.3em]
   &= \max_{\substack{\cA \subset \cZ \\ P_Z(\cA) \geq \delta}} \; \min_{z \in \cA} \; \ell(z).
\end{align*}
Observe that the mapping $\delta \mapsto \bar \varepsilon_Z(\delta)$ is non-increasing and left-continuous.

Heuristically, the difference between $\bar \varepsilon_Z(\delta)$ and $\ubar \varepsilon_Z(\delta)$ can be understood as follows: 
$\ubar \varepsilon_Z(\delta)$ is the smallest upper bound on the worst-case PML over all the \say{good} sets of outputs (i.e., sets with probability at least $1-\delta$). In contrast, $\bar \varepsilon_Z(\delta)$ is the largest lower bound on the worst-case PML over all \say{bad} sets of outputs (i.e., sets with probability at least $\delta$). If $C_Z$ is strictly increasing, then $\ubar \varepsilon_Z(\delta) = \bar \varepsilon_Z(\delta)$ for all $\delta \in (0,1)$, since in that case the inverse of $C_Z$ is well defined and coincides with both the left-continuous and right-continuous quantile functions. More generally, we have $\bar \varepsilon_Z(\delta) \geq \ubar \varepsilon_Z(\delta)$ for all $\delta \in (0,1)$ and the inequality may be strict, particularly when $Z$ is a discrete random variable. 

The difference between $\ubar \varepsilon_Z$ and $\bar \varepsilon_Z$ is further illustrated in the following example.
\begin{example}
Recall $X$ and $Y$ from Example~\ref{ex:prob_pml_not_post_proc}, and fix $\delta = 0.1$. The leakage random variable $\ell(Y)$ takes on two distinct values: a low leakage value of $\log(10/9)$ with probability $0.9$, and a high leakage value of $\log(4)$ with probability $0.1$. For this distribution, we have $\ubar \varepsilon_Y(\delta) = \log (10/9)$, while $\bar \varepsilon_Y(\delta) = \log 4$. Figure~\ref{fig:PML_quantiles_example} shows the distribution function $C_Y$ of $\ell(Y)$, together with two vertical lines marking $\ubar \varepsilon_Y(\delta)$ and $\bar \varepsilon_Y(\delta)$. 

\begin{figure}
    \centering
    \includegraphics[scale=0.8]{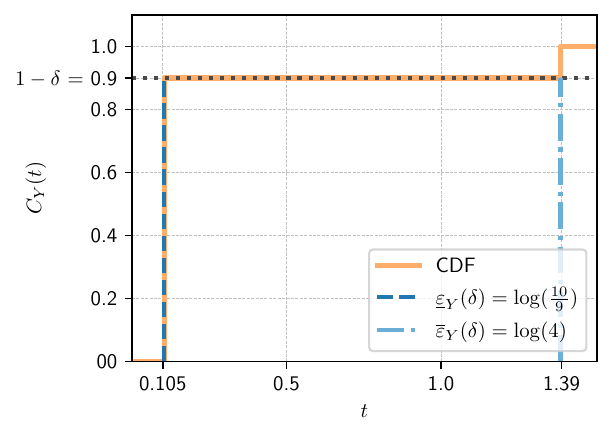}
    \caption{An example of the CDF $C_Y(t)$ together with $\ubar \varepsilon_Y(\delta)$ and $\bar \varepsilon_Y(\delta)$ at $\delta = 0.1$. Since $C_Y$ is not strictly increasing, there is a gap between $\ubar \varepsilon_Y(\delta)$ and $\bar \varepsilon_Y(\delta)$.}
    \label{fig:PML_quantiles_example}
\end{figure}
\end{example}

Below, we show that maximizing over all post-processings of $Y$ eliminates the gap between $\bar \varepsilon_Z(\delta)$ and $\ubar \varepsilon_Z(\delta)$. Thus, both quantities can be used to define the PML envelope.  
\begin{theorem}
\label{thm:comp_eps}
Suppose $X$ and $Y$ are finite random variables. For all $\delta \in (0,1)$, we have
\begin{equation*}
    \varepsilon_c(\delta)
    = \sup_{Z : X - Y - Z} \; \ubar \varepsilon_Z(\delta)
    = \sup_{Z : X - Y - Z} \; \bar \varepsilon_Z(\delta).
\end{equation*}
\end{theorem}
\begin{IEEEproof}
    See Appendix~\ref{sec:thm_comp_eps_proof}.
\end{IEEEproof}

Importantly, Theorem~\ref{thm:comp_eps} explains \emph{why} probabilistic PML is not closed under post-processing. It states that the idea of simply ignoring a small \say{bad} set does not work because, through post-processing, an adversary can increase the information leakage all the way up to the smallest value in the bad set. To illustrate this, let us revisit Example~\ref{ex:prob_pml_not_post_proc}. For $\delta = 0.1$, the good set consists of two outcomes with PML $\log (10/9)$ and the bad set consists of two outcomes with PML $\log(4)$. By Theorem~\ref{thm:comp_eps}, the adversary can increase the leakage from $\log (10/9)$ all the way up to $\log(4)$. Note that by \eqref{eq:pml_bounds}, $\log(4)$ is the largest amount of information \emph{any} mechanism can leak about a uniformly distributed quaternary secret $X$. As $\delta \mapsto \varepsilon_c(\delta)$ is non-increasing, it follows that $\varepsilon_c(\delta) = \log(4)$ for all $\delta \in (0, 0.1]$.

% Theorem~\ref{thm:comp_eps} can also be used to prove that the map $\delta \mapsto \varepsilon_c(\delta)$ is continuous. 
% \begin{cor}
% \label{cor:cont}
% Suppose $X$ and $Y$ are finite random variables. The function $\varepsilon_c$ is continuous on $(0,1)$. 
% \end{cor}

In the following sections, we characterize and bound $\varepsilon_c$. Our analysis frequently relies on Theorem~\ref{thm:comp_eps} and on studying $\bar \varepsilon_Z$ for different choices of $Z$ instead of $\ubar \varepsilon_Z$. We refer to $\bar \varepsilon_Z(\delta)$ as the \emph{PML $\delta$-quantile} of $Z$.

%%%%%%%%%%%%%%%%%%%%%%%%%%%%%%%%%%%%%
%%%%%%%%%%%%%%%%%%%%%%%%%%%%%%%%%%%%%
\subsection{Upper Bounding the PML envelope} 
Computing the PML envelope is, in general, a challenging task, since it involves a maximization over all possible post-processings of $Y$. This motivates us to derive bounds on the envelope. Below, we present a general upper bound.

\begin{theorem}
\label{thm:upper_bound_envelope}
Suppose the joint distribution $P_{XY}$ satisfies $\varepsilon$-PML with $\varepsilon >0$. Then, for all $\delta \in (0,1)$ we have
\begin{equation}
\label{eq:envelope_ml_bound}
    \varepsilon_c(\delta) \leq \min \left \{\cL(X \to Y) + \log \frac{1}{\delta}, \; \varepsilon \right\},
\end{equation}   
where 
\begin{equation*}
    \cL(X \to Y) =  \log \, \bE \big[e^{\ell(Y)}\big] = \log \, \sum_{y \in \cY} \max_{x \in \cX} P_{Y \mid X=x}(y),
\end{equation*}
denotes (the logarithm of) multiplicative Bayes capacity~\cite{alvim2014additive}, also known as maximal leakage~\cite{issaOperationalApproachInformation2020}. 
\end{theorem}
\begin{IEEEproof}
To prove the first term of the bound, fix some $Z$ satisfying the Markov chain $X - Y - Z$, and $t \geq 0$. We write
\begin{subequations}
\begin{align}
    \bP \{\ell(Z) > t\} &= \bP \Big\{e^{\ell(Z)} > e^{t} \Big\} \nonumber\\
    &\leq \bE \big[e^{\ell(Z)}\big] e^{-t} \label{subeq:markov}\\
    &\leq \bE \big[e^{\ell(Y)}\big]  e^{-t}, \label{subeq:data_proc_max_leak}
\end{align}
\end{subequations}
where \eqref{subeq:markov} is due to Markov's inequality and \eqref{subeq:data_proc_max_leak} is due to the data-processing inequality for maximal leakage~\cite[Lemma 1]{issaOperationalApproachInformation2020}. 
Thus, 
\begin{equation*}
    \{t \geq 0 : \bP \{\ell(Z) > t\} \geq \delta \} \subseteq \left\{t \geq 0 : \bE \big[e^{\ell(Y)}\big] e^{-t} \geq \delta \right\}, 
\end{equation*}
which implies that 
\begin{align*}
    \bar \varepsilon_Z(\delta) &= \sup \; \{t \geq 0 : \bP \{\ell(Z) > t\} \geq \delta \}\\
    &\leq \sup \; \left\{t \geq 0 :  \bE \big[e^{\ell(Y)}\big] e^{-t} \geq \delta \right\}\\
    &= \cL(X \to Y)+ \log \Big(\frac{1}{\delta} \Big). 
\end{align*}

To prove the second term, we use the following four facts: 
\begin{enumerate}
    \item $P_{XY}$ satisfies $\varepsilon$-PML if 
    \begin{equation*}
        \max_{y \in \cY} \ell(X \to y) \leq \varepsilon, 
    \end{equation*}

    \item when $\cY$ is a finite set, there exists $\delta_0 >0$ such that for all $0< \delta \leq \delta_0$ we have 
    \begin{equation*}
        \ubar \varepsilon_Y (\delta) = \bar \varepsilon_Y(\delta) = \max_{y \in \cY} \ell(X \to y), 
    \end{equation*}

    \item the post-processing inequality for PML~\cite[Lemma 1]{saeidian2023pointwise_it} states that if the Markov chain $X - Y - Z$ holds, then 
    \begin{equation*}
        \max_{z \in \cZ} \ell_{P_{XZ}}(X \to z) \leq \max_{y \in \cY} \ell_{P_{XY}}(X \to y), 
    \end{equation*}
    and 
    
    \item $\bar \varepsilon_Z(\delta)$, $\ubar \varepsilon_Z(\delta)$, and $\varepsilon_c(\delta)$ are all non-increasing in $\delta$.
\end{enumerate}

Thus, for all $\delta \in (0,1)$, it holds that 
\begin{align*}
    &\varepsilon_c(\delta) \leq \sup_{\delta' \in (0,1)} \varepsilon_c(\delta') = \sup_{Z : X-Y-Z} \; \sup_{\delta' \in (0,1)} \, \bar \varepsilon_Z(\delta')\\
    &= \sup_{Z : X-Y-Z} \; \max_{z \in \cZ} \ell_{P_{XZ}}(X \to z) \leq \max_{y \in \cY} \ell_{P_{XY}}(X \to y)\\
    &\leq \varepsilon. 
\end{align*}
\end{IEEEproof}

We make a few remarks on Theorem~\ref{thm:upper_bound_envelope}. First, when $X$ is finite, \eqref{eq:pml_bounds} implies that PML is uniformly bounded for all outcomes $y \in \cY$. Therefore, regardless of the mechanism $P_{Y\mid X}$, the joint distribution $P_{XY}$ satisfies $\varepsilon$-PML for some finite $\varepsilon$, yielding a valid bound on the envelope. Second, the maximal leakage bound is useful for characterizing the PML envelope under adaptive composition. In particular, suppose we run two mechanisms on $X$, producing outputs $Y_1$ and $Y_2$. Then, using \cite[Corollary 2]{issaOperationalApproachInformation2020} we have
\begin{align*}
    \varepsilon_c(\delta) &\leq \cL(X \to Y_1,Y_2) + \log\frac{1}{\delta}\\
    &\leq \cL(X \to Y_1) + \cL(X \to Y_2 \mid Y_1) + \log\frac{1}{\delta},
\end{align*}
where 
\begin{multline*}
     \cL(X \to Y_2 \mid Y_1) = \\ 
     \max_{\substack{y_1 : \\ P_{Y_1}(y_1) > 0}} \log \, \sum_{y_2} \max_{\substack{x :\\ P_{X \mid Y_1=y_1} (x) >0}} P_{Y_{2} \mid X=x, Y_1 = y_1}(y_2), 
\end{multline*}
is the conditional form of maximal leakage \cite[Thm. 6]{issaOperationalApproachInformation2020}. When the composition is non-adaptive, the mechanism releasing $Y_2$ does not depend on $Y_1$, so $P_{Y_2 \mid X, Y_1} = P_{Y_2 \mid X}$. In that case, $\cL(X \to Y_2 \mid Y_1) \leq \cL(X \to Y_2)$ and the composition bound simplifies to 
\begin{align*}
    \varepsilon_c(\delta) \leq \cL(X \to Y_1) + \cL(X \to Y_2) + \log\frac{1}{\delta}.
\end{align*}
By induction, if we adaptively run $n$ mechanisms on $X$ and obtain outputs $Y^n=(Y_1,\dots,Y_n)$, then we have
\begin{align*}
    \varepsilon_c(\delta) &\leq \cL(X \to Y^n) + \log\frac{1}{\delta}\\
    &\leq \sum_{i=1}^n \cL(X \to Y_i \mid Y^{i-1}) + \log\frac{1}{\delta}.
\end{align*}
Third, both terms of the upper bound are themselves robust under post-processing.

\subsection{Lower Bounding the PML envelope} 
To obtain lower bounds on $\varepsilon_c$, we may restrict attention to specific classes of post-processings and compute the PML $\delta$-quantile of the resulting outputs. This can be viewed as restricting the computational power of the adversary. In general, the choice of post-processings used to derive lower bounds is mechanism dependent (see Section~\ref{ssec:rr} for an illustration). Nevertheless, it is instructive to consider two simple and broadly applicable instances: (i) taking $Z = Y$, i.e., no post-processing, which yields $\bar{\varepsilon}_Y$; and (ii) restricting attention to binary post-processings of $Y$.

Let
\begin{equation}
\label{eq:binary_envelope}
    \varepsilon_b(\delta) \coloneqq \sup_{Z : X - Y - Z,\; \mathcal{Z} = \{0,1\}} \bar \varepsilon_Z(\delta), \quad \delta \in (0,1),
\end{equation}
denote the \emph{binary} envelope of $P_{XY}$. Then, for all $\delta \in (0,1)$, we have the lower bound
\begin{equation}
\label{eq:envelope_lower_bound}
    \varepsilon_c(\delta) \geq
    \max \big\{\bar \varepsilon_Y(\delta), \varepsilon_b(\delta) \big\}.
\end{equation}
One may obtain sharper lower bounds by extending the bound to ternary, quaternary, and other higher-order $k$-ary post-processings. As we will see below, the advantage of $\varepsilon_b$ is that it is simple to compute. 

The quantity $\varepsilon_b$ previously appeared in the work of \textcite{saeidian2023pointwise_it} (in a slightly different form) as a stand-alone privacy definition. In contrast, here, we use $\varepsilon_b$ only as a lower bound on the PML envelope. Algorithm~\ref{alg:epsb} provides a procedure for computing $\varepsilon_b$ based on the proof of \cite[Thm. 3]{saeidian2023pointwise_it}. The main idea underlying both this proof and Algorithm~\ref{alg:epsb} is to generalize the notion of PML from individual outcomes to \emph{events}. In particular, given an event $\mathcal E \subseteq \mathcal Y$ with $P_Y(\mathcal E)>0$, and $Z=\ind_{\mathcal E}(Y)$, we may define event-wise leakage as
\begin{equation*}
    \ell_{P_{XY}}(X\to \mathcal E)\coloneqq \ell_{P_{XZ}}(X\to 1).
\end{equation*}
Technically, this is not a new concept since event-wise leakage is simply the PML of the affirmative outcome of an indicator function. Nevertheless, it provides a convenient shorthand for reasoning about post-processing. Note that extending PML to events is natural in this context, since any binary post-processing corresponds to selecting an event in $\mathcal Y$ and revealing whether or not the outcome $Y$ lies in that event.

Heuristically, Algorithm~\ref{alg:epsb} works as follows: For each $x\in\mathcal X$, we compute the largest value of
\begin{equation*}
    \frac{P_{Y\mid X=x}(\mathcal E)}{P_Y(\mathcal E)},
\end{equation*}
over events $\mathcal E\subseteq\mathcal Y$ with $P_Y(\mathcal E)=\delta$ (with possible randomization at the boundary).\footnote{Here, we use the term \say{event} in a generalized sense that allows randomization. In particular, each outcome $y \in \cY$ may be in the set $\cE$ with a certain probability. See Appendix~\ref{sec:events} for details.}
The quantity $\varepsilon_b$ is then obtained by taking the logarithm of the maximum of this value over all $x\in\mathcal X$. Further discussion about the information leakage of events and their connection to post-processing is provided in Appendix~\ref{sec:events}. 

Finally, we emphasize that both terms $\bar \varepsilon_Y$ and $\varepsilon_b$ must be included in the lower bound \eqref{eq:envelope_lower_bound}. Depending on the setup, either term may yield the tighter bound.

\begin{algorithm}[t]
\caption{Computing $\varepsilon_b(\delta)$ for a joint distribution $P_{XY}$}
\label{alg:epsb}
\begin{algorithmic}[1]
\Require $P_{Y \mid X}$, $P_X$ and $\delta\in(0,1)$
\Ensure $\varepsilon_b(\delta)$
\State Compute the marginal $P_Y = P_{Y \mid X} \circ P_X$
\State Initialize $M\gets 0$
\ForAll{$x\in\cX$}
    \ForAll{$y\in\cY$}
    \State $s_x(y)\gets \frac{P_{Y \mid X=x}(y)}{P_Y(y)}$
    \EndFor
    \State Sort $\cY$ as $(y_1,\dots,y_{\abs{\cY}})$ so that 
    \begin{equation*}
        s_x(y_1)\ge s_x(y_2)\ge\cdots\ge s_x(y_{\abs{\cY}})
    \end{equation*}
    \State Find the smallest index $k^\star$ such that $\sum_{j=1}^{k^\star} P_Y(y_j)\ge \delta$
    \State $p\gets \sum_{j=1}^{k^\star-1} P_Y(y_j)$
    \Statex \vspace{-0.6em}
    \State $\zeta \gets \frac{\delta-p}{P_Y(y_{k^\star})}$
    \Statex \vspace{-0.5em}
    \State $v \gets \frac{1}{\delta}\left(\sum_{j=1}
    ^{k^\star-1} P_{Y \mid X=x}(y_j) +  \zeta P_{Y \mid X=x}(y_{k^\star}) \right)$
    \Statex \vspace{-0.6em}
    \State $M\gets \max (M,\,v)$
\EndFor
\State \Return $\varepsilon_b \gets \log M$
\end{algorithmic}
\end{algorithm}

\section{Applications}
\label{sec:examples}
We now calculate and bound the PML envelope in two canonical settings. These are: the \emph{PML–extremal mechanisms} in the high-privacy regime~\cite{grosseExtremalMechanismsPointwise2024} and the \emph{randomized response} mechanism~\cite{warnerRandomizedResponseSurvey1965,kairouzExtremalMechanismsLocal2016}.

\subsection{PML-extremal Mechanisms}
In \cite{grosseExtremalMechanismsPointwise2024}, Grosse \emph{et al.} studied the design of optimal privacy mechanisms under the $\varepsilon$–PML constraint for a broad class of convex utility functions. The resulting mechanisms were termed \emph{PML–extremal mechanisms}. Fix a prior distribution $P_X$ on $\mathcal{X} = [k]$ with $k \geq 2$. \textcite{grosseExtremalMechanismsPointwise2024} showed that in the \emph{high-privacy regime}, corresponding to
\begin{equation*}
    0 < \varepsilon < \log \left(\frac{1}{1 - \min_{x \in \mathcal{X}} P_X(x)}\right),
\end{equation*}
the optimal mechanism has the form
\begin{equation*}
     P^*_{Y \mid X=i}(j)
     = \begin{cases}
        1 - e^{\varepsilon}\bigl(1 - P_X(i)\bigr), & \text{if } i=j, \\[0.25em]
        e^{\varepsilon} P_X(j), & \text{if } i \neq j,
    \end{cases}
\end{equation*}
where $\mathcal{Y} = \mathcal{X}$, and $i,j \in [k]$. Note that the outcomes of this mechanism have the PML $\ell(X \to j) = \varepsilon$ for all $j \in [k]$. 

\begin{theorem}
\label{thm:pml_ext}
Let $P_X$ be a distribution on $\mathcal{X} = [k]$, and fix $0 < \varepsilon < - \log \left(1 - \min_{x \in \mathcal{X}} P_X(x)\right)$. Let $P^*_{Y \mid X}$ denote the PML-extremal mechanism. Then, for all $\delta \in (0,1)$, the PML envelope is $\varepsilon_c(\delta) = \varepsilon$.
\end{theorem}

\begin{IEEEproof}
By Theorem~\ref{thm:upper_bound_envelope}, any mechanism satisfying $\varepsilon$-PML has $\varepsilon_c (\delta) \leq \varepsilon$ for all $\delta \in (0,1)$. We argue that $P^*_{Y \mid X}$ achieves this upper bound.

Let $\cA \subset \cY$ be an arbitrary set with probability $P_{Y}(\mathcal{A}) \geq \delta$. Since every $j \in {\cY}$ has leakage $\varepsilon$, we have $\min_{j \in \mathcal{A}} \ell(X \to j) = \varepsilon$. Using \eqref{eq:envelope_lower_bound} we have
\begin{align*}
\varepsilon_c(\delta) \geq \bar \varepsilon_{Y}(\delta) &= \max_{\substack{\cE \subset {\cY} \\ P_{Y}(\cE) \geq \delta}} \min_{j \in \mathcal{E}} \ell(X \to j)\\
&\geq \min_{j \in \mathcal{A}} \ell(X \to j) = \varepsilon.
\end{align*}
Combining this with the upper bound establishes the claim.
\end{IEEEproof}

The proof of Theorem~\ref{thm:pml_ext} shows that if $\ell(X \to y)=\varepsilon$ for all $y$, then the PML envelope is constant in $\delta$. We state this in the following corollary.
\begin{cor}
\label{cor:constant_pml}
Suppose $\ell(X \to y)=\varepsilon$ for all outcomes $y \in \mathcal{Y}$ with $\varepsilon>0$. Then, the PML envelope satisfies $\varepsilon_c(\delta)=\varepsilon$ for all $\delta \in (0,1)$.
\end{cor}

%%%%%%%%%%%%%%%%%%%%%%%%%%%%%%
%%%%%%%%%%%%%%%%%%%%%%%%%%%%%%
\subsection{Randomized Response Mechanism} 
\label{ssec:rr}
Given an integer $k \geq 2$, let $\cX = \cY = [k]$. The \emph{$k$-randomized response} ($k$-RR) mechanism with parameter $\varepsilon_r > 0$ is defined as 
\begin{equation}
    P_{Y \mid X=i}(j) = 
    \begin{cases}\frac{e^{\varepsilon_r}}{e^{\varepsilon_r} + k-1},& j=i, \\[0.4em]
        \frac{1}{e^{\varepsilon_r} + k-1}, & j\neq i,
    \end{cases} \quad i,j \in [k].
\end{equation}
The $k$-RR mechanism satisfies $\varepsilon_r$-LDP. For simplicity, let 
\begin{gather*}
    \alpha \coloneqq \frac{e^{\varepsilon_r}}{e^{\varepsilon_r} + k-1},\\
    \beta \coloneqq \frac{1}{e^{\varepsilon_r} + k-1},\\
    p_i \coloneqq P_X(i), \quad i \in [k],\\
    q_j \coloneqq P_Y(j), \quad j \in [k], 
\end{gather*}
and observe that $q_i = \beta + (\alpha - \beta) p_i$ for all $i \in [k]$. Furthermore, each outcome of the mechanism has the PML 
\begin{equation*}
    \ell(X \to j) = \log \frac{\alpha}{q_j}, \quad j \in [k].
\end{equation*}
Without loss of generality, let us assume that $p_1 \leq p_2 \leq \dots \leq p_k$, which implies that $q_1 \leq q_2 \leq \dots \leq q_k$ and $\ell(X \to 1) \geq \ell(X \to 2) \geq \dots \geq \ell(X \to k)$. 

\begin{theorem}
\label{thm:k-rr}
Let $P_X$ be a distribution on $\mathcal{X} = [k]$ and let $P_{Y \mid X}$ be the $k$-RR mechanism. The PML envelope satisfies the following: 
\begin{enumerate}
    \item If $\delta \in (0,q_1]$, then $\varepsilon_c(\delta) = \log \frac{\alpha}{q_1}$.
    \item For $\delta \in (q_1, 1)$, the PML envelope admits the upper bound  
    \begin{equation*}
        \varepsilon_c(\delta) \leq\min \left \{ \log \frac{k \alpha}{\delta},  \log \frac{\alpha}{q_1} \right\}.
    \end{equation*}

    \item For $\delta \in (q_1,1)$, let $N \in \{2,\dots,k\}$ be the unique index such that $\sum_{j=1}^{N-1} q_j < \delta \leq \sum_{j=1}^N q_j$. Let $\theta := \frac{\delta - \sum_{j=1}^{N-1} q_j}{q_N} \in (0,1]$, and suppose the prior satisfies 
    \begin{equation*}
        p_{N} \leq \frac{\alpha \sum_{j=1}^{N-1} p_j + \beta}{(N-2)\alpha + \beta}.
        \end{equation*}
    Then, the PML envelope admits the lower bound 
    \begin{equation*} 
    \varepsilon_c(\delta) \geq h_\delta(\theta),  
    \end{equation*}
    where 
    \begin{equation*}
        h_\delta(\theta) = \begin{cases}
            \ell(X \to N-1) & \text{ if }\, 0 < \theta \leq \theta_1, \\[0.4em] 
            \log \; \frac{(N-1) \alpha + \theta \beta}{\delta} & \text{ if }\, \theta_1 < \theta \leq  \theta_2, \\[0.4em]
            \ell(X \to N) & \text{ if }\, \theta_2 < \theta \leq 1, 
        \end{cases}
    \end{equation*}
    and 
    \begin{gather*}
        \theta_1 = \frac{\alpha \Big( (N-2) q_{N-1} - \sum_{j=1}^{N-2} q_j \Big)}{\alpha q_N - \beta q_{N-1}},\\ 
        \theta_2 = \frac{\alpha \Big( (N-1) q_{N} - \sum_{j=1}^{N-1} q_j \Big)}{q_N (\alpha -\beta)}. 
    \end{gather*}
\end{enumerate}
\end{theorem}
\begin{IEEEproof}
 \begin{enumerate}
    \item The $k$-RR mechanism satisfies $\log \frac{\alpha}{q_1}$-PML, so we have $\varepsilon_c(\delta) \leq \log \frac{\alpha}{q_1}$ which holds for all $\delta \in (0,1)$. Suppose $\delta \in (0,q_1]$, and let $\cA = \{1\}$ which satisfies $P_Y(\cA) = q_1$. Then, we have 
    \begin{align*}
        \varepsilon_c(\delta) \geq \bar \varepsilon_{Y}(\delta) &= \max_{\substack{\cE \subset {\cY} \\ P_{Y}(\cE) \geq \delta}} \min_{j \in \mathcal{E}} \, \ell(X \to j)\\
        &\geq \min_{j \in \mathcal{A}} \ell(X \to j) = \log \frac{\alpha}{q_1}.
    \end{align*}

    \item We use Theorem~\ref{thm:upper_bound_envelope} and calculate the maximal leakage of the $k$-RR mechanism: 
    \begin{equation*}
        \cL(X \to Y) = \log \, \sum_{j \in [k]} \max_{i \in [k]} P_{Y \mid X=i}(j) = \log (k\alpha).  
    \end{equation*}
    This yields the upper bound
    \begin{equation*}
        \varepsilon_c(\delta) \leq \cL(X \to Y) + \log \frac{1}{\delta} = \log \frac{k \alpha}{\delta}. 
    \end{equation*}

    \item First, observe that \eqref{eq:envelope_lower_bound} yields 
    \begin{equation*}
        \varepsilon_c(\delta) \geq \bar \varepsilon_Y(\delta) = \log \frac{\alpha}{q_N} = \ell(X \to N),
    \end{equation*}
    for $\sum_{j=1}^{N-1} q_j < \delta \leq \sum_{j=1}^N q_j$. Below, we construct a post-processing $Z$ and tighten the lower bound to 
    \begin{equation*}
         \varepsilon_c(\delta) \geq \max \{\bar \varepsilon_Y(\delta),  \bar \varepsilon_Z(\delta)\}.  
    \end{equation*}

    % Given $N$ and $\theta$, let $P_{W \mid Y}$ be a channel with output alphabet $\cW = [N+1]$ expressed as 
    % \begin{equation*}
    %     P_{W\mid Y=j}(w) =
    %     \begin{cases}
    %         1, & \text{if } j \in [N-1] \text{ and } w = j, \\[0.25em]
    %         \theta, & \text{if } j = w = N, \\[0.25em]
    %         1 - \theta, & \text{if } j = N \text{ and } w= N+1, \\[0.25em]
    %         1, & \text{if } j > N \text{ and } w = N+1, \\[0.25em]
    %         0, & \text{otherwise.}
    %     \end{cases}
    % \end{equation*}
    % Consider the subset $\cA = [N]$ of outcomes of $W$ and observe that
    % \begin{equation*}
    %     P_W(\cA) = \sum_{w =1}^N P_W(w) = \sum_{w=1}^{N-1} q_w + \theta q_N = \delta. 
    % \end{equation*}
    % In addition, each $w \in \cA$ has PML 
    % \begin{equation*}
    %     \ell_{P_{XW}}(X \to w) = \log\frac{\max\limits_x P_{W\mid X}(w\mid x)}{P_Z(w)} = \log \frac{\alpha}{q_w}.
    % \end{equation*}
    % Thus, using $W$, we obtain the lower bound 
    % \begin{equation*}
    %     \varepsilon_c(\delta) = \bar \varepsilon_W(\delta) \geq \min_{w \in \cA} \ell_{P_{XW}}(X \to w) = \log \frac{\alpha}{q_N}.
    % \end{equation*}
    
    Let $\mathcal{Z} = [N]$ and $\eta = (\eta_1,\dots,\eta_{N})$ be a tuple satisfying 
    \begin{gather}
        \eta_z \geq 0, \quad \text{for all } z \in [N],\nonumber \\
        \sum_{z=1}^{N-1} \eta_z = \theta, \quad \eta_{N} = 1 - \theta. \label{eq:eta_constraint}
    \end{gather}
    Define
    \begin{equation*}
        P_{Z\mid Y=j}(z) =
        \begin{cases}
            1, & \text{if } j \in [N-1] \text{ and } z = j, \\[0.25em]
            \eta_z, & \text{if } j = N \text{ and } z \in [N], \\[0.25em]
            1, & \text{if } j > N \text{ and } z = N, \\[0.25em]
            0, & \text{otherwise.}
        \end{cases}
    \end{equation*}
    In words, each $k$-RR outcome $j \in [N-1]$ is deterministically mapped to $z \in [N-1]$, all $j > N$ are mapped to a single catch-all symbol $z = N$, and $j = N$ is split across $z = 1, \dots, N$ with weights $\eta_1,\dots,\eta_{N}$.

    Consider the event $\cB = [N-1]$ of outcomes of $Z$, and note that by construction, we have 
    \begin{align*}
        P_Z(\cB) &= \sum_{z=1}^{N-1} P_Z(z)\\
        &= \sum_{z=1}^{N-1} q_z + \eta_z q_N \\
        &= \sum_{z=1}^{N-1} q_z + \theta q_N = \delta.        
    \end{align*}
    Thus, we may use the set $\cB$ to obtain lower bounds on the PML envelope: 
    \begin{equation*}
        \varepsilon_c(\delta) = \bar \varepsilon_Z(\delta) \geq \min_{z \in \cB} \ell_{P_{XZ}}(X \to z). 
    \end{equation*}
    Our goal is to optimize the weights $\{\eta_i\}_{i=1}^N$ in order to obtain the tightest possible lower bound.  Note that if $N=2$, then \eqref{eq:eta_constraint} forces $\eta_1 = \theta$, so $P_{Z \mid Y}$ is fully specified. Therefore, for the rest of the proof assume that $N>2$. 
    
    We begin by calculating the PML for symbols in the set $\cB$. For each $i \in [N-1]$, using the structure of the $k$-RR mechanism and $P_{Z \mid Y}$, we observe that  
    \begin{gather*}
        \max_x P_{Z\mid X}(i\mid x) = P_{Z\mid X}(i\mid i) = \alpha + \eta_i \beta,\\
         P_Z(i) = q_i + \eta_i q_N.
    \end{gather*}
    Hence, the PML is 
    \begin{equation*}
        \ell_{P_{XZ}}(X \to i) = \log\frac{\max\limits_x P_{Z\mid X}(i\mid x)}{P_Z(i)} = \log \frac{\alpha + \eta_i \beta}{q_i + \eta_i q_N},
    \end{equation*}     
    Let 
    \begin{equation*}
        M_i(\eta_i) \coloneqq \frac{\alpha + \eta_i \beta}{q_i + \eta_i q_N}, \quad i=1, \dots, N-1, 
    \end{equation*}    
    and consider the optimization problem:  
    \begin{gather*}
        \max_{\eta_1, \dots, \eta_{N-1}} \, \min_{i \in [N-1]} M_i(\eta_i),\\
        \text{subject to } \sum_{i=1}^{N-1} \eta_i = \theta,\\
        \eta_i \geq 0, \; i=1,\dots,N-1. 
    \end{gather*}
    
    It is easy to verify that $M_i''(\eta_i) \geq 0$, so $M_i$ is convex. Thus, the above optimization problem is \emph{not} a convex one (since the minimum of a collection of convex functions need not be convex). Nevertheless, we can solve it by inspection.

    \smallskip
    \textbf{First regime.}
    Let us start by noting that $M_i'(\eta_i) < 0$ for $\eta_i \geq 0$, implying that $M_i(\eta_i) \leq M_i(0) = \frac{\alpha}{q_i}$ for all $i \in [N-1]$. This yields the upper bound on the objective
    \begin{equation}
    \label{eq:rr_opt_upper_bound}
        \max_{\eta_1, \dots, \eta_{N-1} \geq 0} \min_{i \in [N-1]} M_i(\eta_i) \leq \min_{i\in [N-1]} \frac{\alpha}{q_i} = \frac{\alpha}{q_{N-1}}, 
    \end{equation}
    since $q_1 \leq \dots \leq q_k$. This bound is achievable at $\theta = 0$ since $\eta_1 = \dots = \eta_{N-1} = 0$ is feasible at this point.\footnote{Technically, we assume that $\theta>0$, but we may consider the limiting value of the objective as $\theta \downarrow 0$ since the $M_i$'s are continuous.} Next, we argue that there exists $\theta_1 \geq 0$ such that the upper bound in~\eqref{eq:rr_opt_upper_bound} is achievable for $\theta \in (0,\theta_1]$. This is because in order to achieve $\frac{\alpha}{q_{N-1}}$, all we need is to have $\eta_{N-1} = 0$ and $M_i(\eta_i) \geq \frac{\alpha}{q_{N-1}}$ for $i=1, \dots, N-2$. 

    Let $\eta_i^*$ be such that $M_i(\eta_i^*) = \frac{\alpha}{q_{N-1}}$, that is, 
    \begin{align*}
        &M_i(\eta_i^*) = \frac{\alpha + \eta_i^* \beta}{q_i + \eta_i^* q_N} = \frac{\alpha}{q_{N-1}} \Longleftrightarrow\\
        &\eta^*_i =  \frac{\alpha (q_{N-1} - q_i)}{\alpha q_N - \beta q_{N-1}} \geq 0, \quad i=1, \dots, N-2,   
    \end{align*}
    and also $\eta_{N-1}^* = 0$. This choice of the parameters yields
    \begin{equation*}
        \theta_1 = \sum_{i=1}^{N-1} \eta_i^* = \frac{\alpha \Big( (N-2) q_{N-1} - \sum_{i=1}^{N-2} q_i \Big)}{\alpha q_N - \beta q_{N-1}} \geq 0. 
    \end{equation*}
    Therefore, assuming that $\theta_1 > 0$, for $\theta \in (0,\min \{\theta_1,1\}]$, we have the first piece of the lower bound 
    \begin{align*}
        \varepsilon_c(\delta) &\geq \min_{z \in \cA} \ell_{P_{XZ}}(X \to z)\\
        &= \min_{i \in [N-1]} \log M_i(\eta_i^*) = \log \frac{\alpha}{q_{N-1}}.  
    \end{align*}
    Note that in this regime, we use the lower bound obtained from $Z$ and not $Y$ itself, since $\frac{\alpha}{q_{N-1}} \geq \frac{\alpha}{q_{N}}$.  
    
    %%%%%%%%%%%%%%%%%%%%%%%%%%%%%%%
    \smallskip
    \textbf{Second regime.}
    Next, suppose $\theta_1 <1$ and $\theta > \theta_1$. In the second regime, we are forced to increase at least one $\eta_i$ beyond $\eta_i^*$, so the objective falls below $\frac{\alpha}{q_{N-1}}$. Let $\{\tilde \eta_i\}$ denote the optimal parameters. There exists a common threshold $\tau \in [\frac{\alpha + \beta}{q_1 + q_N}, \frac{\alpha}{q_{N-1}}]$ such that 
    \begin{align*} 
    M_i(\tilde \eta_i) = \frac{\alpha + \tilde \eta_i \beta}{q_i + \tilde \eta_i q_N} = \tau \Longleftrightarrow \tilde \eta_i(\tau) = \frac{\alpha - \tau q_i}{\tau q_N - \beta} > 0, 
    \end{align*}
    for $i=1, \dots, N-1$.\footnote{Note that each $M_i$ is continuous and strictly decreasing. Therefore, if $M_i(\eta_i) > M_j( \eta_j)$, then there exists $\zeta >0$ such that $M_j( \eta_j) < M_j( \eta_j - \zeta) = M_i( \eta_i + \zeta)$. Thus, the optimal parameters must yield a common value for all $M_i$'s.} Hence, $\theta$ can be expressed as 
    \begin{align*}
        \theta(\tau) = \sum_{i=1}^{N-1} \tilde \eta_i(\tau) = \frac{(N-1) \alpha - \tau \sum_{i=1}^{N-1} q_i}{\tau q_N - \beta}. 
    \end{align*}
    Solving for $\tau$ gives 
    \begin{equation*}
        \tau(\theta) = \frac{(N-1) \alpha + \theta \beta}{\sum_{i=1}^{N-1} q_i + \theta q_N} = \frac{(N-1) \alpha + \theta \beta}{\delta}, \quad \theta >  \theta_1. 
    \end{equation*}
    Therefore, for $\theta \in (\theta_1,\theta_2]$ (with $\theta_2$ specified below), we get the middle piece of the lower bound  
    \begin{equation}
    \label{eq:lower_bound_Z}
        \varepsilon_c(\delta) \geq \min_{i \in [N-1]} \log M_i(\tilde \eta_i) = \log \frac{(N-1) \alpha + \theta \beta}{\delta}.
    \end{equation}

    %%%%%%%%%%%%%%%%%%%%%%%%%%%%%%%
    \textbf{Third regime.}
    The point $\theta_2$ is the value where any further increase in $\theta$ would make the lower bound in~\eqref{eq:lower_bound_Z} drop below $\bar \varepsilon_Y(\delta)$, i.e.,
    \begin{equation*}
        \tau(\theta_2) = \frac{\alpha}{q_N} \;\; \Longleftrightarrow\;\;
        \theta_2 = \frac{\alpha\bigl( (N-1)q_N - \sum_{i=1}^{N-1} q_i\bigr)}
               {q_N(\alpha - \beta)}.
    \end{equation*}
    Thus, in the third regime $\theta \in (\theta_2, 1]$ we use
    \begin{equation*}
         \varepsilon_c(\delta) \geq \bar \varepsilon_Y(\delta) = \log \frac{\alpha}{q_N}. 
    \end{equation*}

    \begin{figure*}[t]
     \centering
     \begin{subfigure}[b]{0.9\columnwidth}
         \centering
         \includegraphics[width=\linewidth]{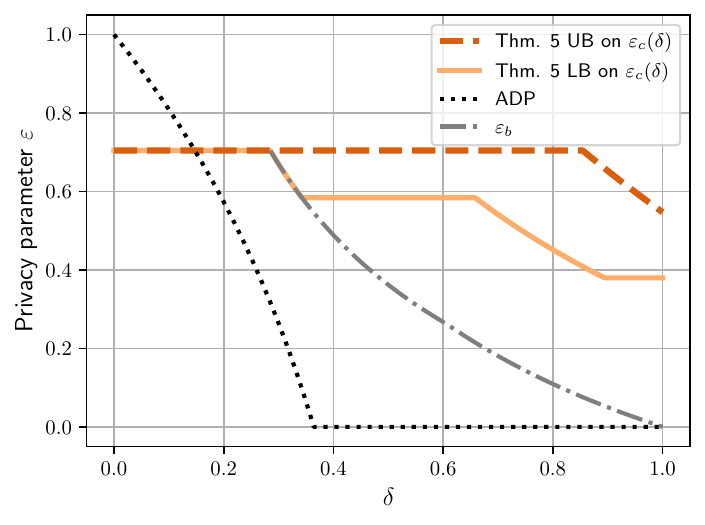}
         \caption{$k=3$ and $P_X = (0.2, 0.3, 0.5)$.}
         \label{fig:k3}
     \end{subfigure}
     \hfill
     \begin{subfigure}[b]{0.9\columnwidth}
         \centering
         \includegraphics[width=\linewidth]{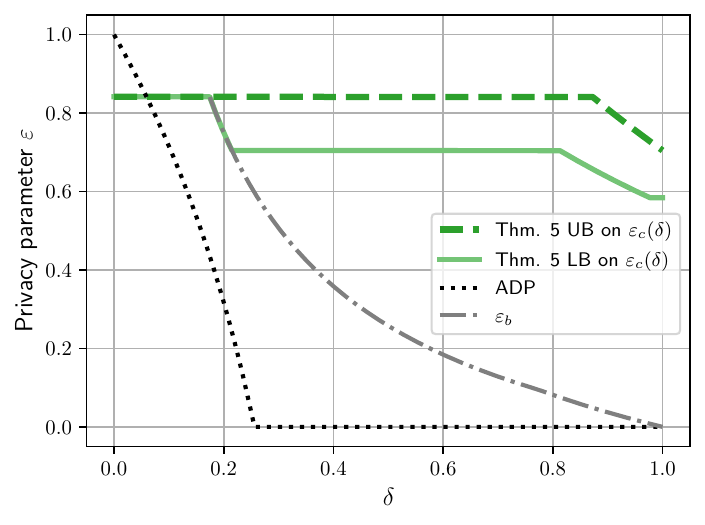}
         \caption{$k=5$ and $P_X = (0.1, 0.2, 0.2, 0.2, 0.3)$.}
         \label{fig:k5}
     \end{subfigure}
     
     \vspace{0.5em}
     \begin{subfigure}[b]{0.9\columnwidth}
         \centering
         \includegraphics[width=\linewidth]{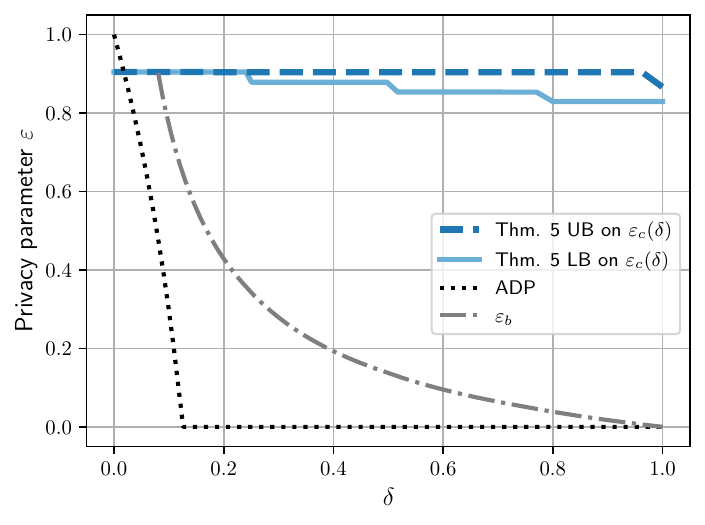}
         \caption{$k=12$ and $P_X$ is the four-level prior with $\rho=0.1$.}
         \label{fig:k12}
     \end{subfigure}
     \hfill
     \begin{subfigure}[b]{0.9\columnwidth}
         \centering
         \includegraphics[width=\linewidth]{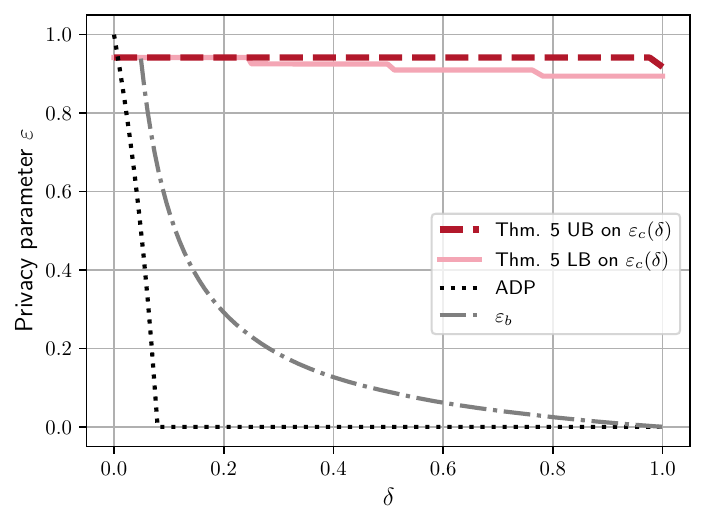}
         \caption{$k=20$ and $P_X$ is the four-level prior with $\rho=0.1$.}
         \label{fig:k20}
     \end{subfigure}
        \caption{PML envelope upper bounds (UBs), lower bounds (LBs), $\varepsilon_b$, and $(\varepsilon_{\mathrm{DP}}(\delta),\delta)$-DP guarantees for the $k$-RR mechanism with $\varepsilon_r = 1.0$. Each sub-figure illustrates a different alphabet size $k$ and prior $P_X$.}
        \label{fig:krr}
\end{figure*}
    
    Finally, we find the condition on the prior distribution ensuring that $\theta_2 \leq 1$. Two conditions need to be satisfied for this: We require $\theta_1 \leq 1$ (ensuring that we enter the second regime) and also $\tau(1) \leq \frac{\alpha}{q_N}$ (ensuring that we enter the third regime). By using $q_i = \beta + (\alpha - \beta) p_i$, and after some algebra, we obtain the following conditions: 
    \begin{gather}
        p_{N-1} \leq \frac{\alpha \sum_{i=1}^N p_i + \beta}{(N-1) \alpha + \beta}, \label{eq:cond_1}\\
        p_{N} \leq \frac{\alpha \sum_{i=1}^{N-1} p_i + \beta}{(N-2) \alpha + \beta}.\label{eq:cond_2}
    \end{gather}
    Observe that \eqref{eq:cond_2} can be written in the form 
    \begin{equation*}
        (N-2) \alpha + \beta \leq \frac{\alpha \sum_{i=1}^{N-1} p_i + \beta}{p_N},
    \end{equation*}
    which implies that 
    \begin{align*}
        \frac{\alpha \sum_{i=1}^N p_i + \beta}{(N-1) \alpha + \beta} \geq \frac{\alpha \sum_{i=1}^N p_i + \beta}{\alpha + \frac{\alpha \sum_{i=1}^{N-1} p_i + \beta}{p_N}} = p_N \geq p_{N-1},
    \end{align*}
    so if \eqref{eq:cond_2} is satisfied, \eqref{eq:cond_1} is also automatically satisfied. 
\end{enumerate}
\end{IEEEproof}

The lower bounds in Theorem~\ref{thm:k-rr} are stated under the assumption that $0<\theta_1<\theta_2$. The following degenerate cases are handled by interpreting the piecewise expression in the natural way: (i) If $N=2$ or $q_1=\cdots=q_{N-1}$, then $\theta_1=0$, so the first regime is vacuous and omitted; (ii) If $q_{N-1}=q_N$, then $\theta_1=\theta_2$, so the second regime is vacuous and omitted. In this case, the bound equals $\ell(X\to N-1)=\ell(X\to N)$ for all $\theta\in(0,1]$. Note that if $P_X$ is uniformly distributed, then $\ell(X \to j) = \log(k\alpha)$ for all $j \in [k]$, and applying Corollary~\ref{cor:constant_pml} gives $\varepsilon_c(\delta) = \log(k\alpha)$ for all $\delta \in (0,1)$.

The condition on the prior in Theorem~\ref{thm:k-rr} ensures that all three segments of $h_\delta$ can be active, i.e., it guarantees that $\theta_2 \leq 1$. We conjecture that under this condition, $h_\delta$ coincides with $\varepsilon_c$, that is, the construction used in the proof of Theorem~\ref{thm:k-rr} is in fact optimal. Establishing this conjecture would require either proving the optimality of the construction or deriving a matching upper bound. We leave this as an open problem for future work.

Figure~\ref{fig:krr} illustrates the bounds in Theorem~\ref{thm:k-rr}. For comparison, we also plot $\varepsilon_b$ and the ADP curve of the $k$-RR mechanism. Recall from \cite[Thm.~2]{ballePrivacyProfilesAmplification2020a} that, for fixed $\delta \in (0,1)$, the $k$-RR mechanism satisfies $(\varepsilon_{\text{DP}}(\delta), \delta)$-DP with 
\begin{equation*}
    \varepsilon_{\text{DP}}(\delta) = \max \big\{\log \frac{\alpha - \delta}{\beta}, 0 \big\}.
\end{equation*}
Moreover, for $\delta \in (q_1, 1)$, the binary envelope $\varepsilon_b(\delta)$ is derived in Appendix~\ref{ssec:binary_envelope_krr} and is given by
\begin{equation*}
    \varepsilon_b(\delta) = \log \frac{\alpha + (N-2 + \theta) \beta}{\delta}.
\end{equation*}
Here, we consider the $k$-RR mechanism with $\varepsilon_r = 1.0$ for $k \in \{3, 5, 12, 20 \}$ under various prior distributions. To simplify the specification of priors over larger alphabets, we use a \emph{four-level} construction obtained by partitioning the alphabet into four blocks of equal size (assuming $k$ is divisible by $4$). Fix $\rho \in (0,1/3)$ and define
\begin{equation*}
    m_0 = 1-3\rho,\quad
    m_1 = 1-\rho,\quad
    m_2 = 1+\rho,\quad
    m_3 = 1+3\rho.
\end{equation*}
Then, set 
\begin{equation*}
    p_i = \frac{m_{r(i)}}{k}, \quad r(i) \coloneqq \left\lfloor \frac{4 (i-1)}{k} \right \rfloor,
\end{equation*}
for $i \in [k]$. With this construction, $P_X$ is constant on blocks of size $\frac{k}{4}$, and satisfies $\sum_{i=1}^k p_i=1$.

We make the following observations from Figure~\ref{fig:krr}. 
First, the binary envelope $\varepsilon_b(\delta)$ overlaps with the lower bound in Theorem~\ref{thm:k-rr} when $\delta$ is slightly above $q_1$, i.e., when $N=2$ and $\theta \leq \theta_2$ (assuming that $q_1 < q_2)$. Outside this regime, $\varepsilon_b(\delta)$ can be a significantly looser lower bound than $h_\delta$, with the gap becoming more pronounced as $k$ increases. 
Second, when $\delta$ is small, $\log(\alpha/q_1)$ is the tighter upper bound on $\varepsilon_c$ but as $\delta$ increases, the maximal-leakage upper bound $\log(k\alpha/\delta)$ becomes tighter. 
Third, as $k$ grows, the gap between the upper and lower bounds in Theorem~\ref{thm:k-rr} shrinks.

Finally, the PML envelope is qualitatively and quantitatively different from ADP. As $\delta$ varies from $0$ to $1$, $\varepsilon_{\text{DP}}(\delta)$ ranges from $\varepsilon_r$ down to $0$. In contrast, $\varepsilon_c(\delta)$ equals $\log(\alpha/q_1)$ as $\delta \to 0$ (which is strictly smaller than $\varepsilon_r$) and can remain bounded away from $0$ as $\delta \to 1$. We also observe that as $k$ increases, $\varepsilon_{\text{DP}}(\delta)$ decreases more quickly in $\delta$, implying that $k$-RR provides strong DP guarantees even for small $\delta$.  This behavior is not observed for the PML envelope: even for large $k$, $\varepsilon_c(\delta)$ may decrease only mildly with $\delta$. We emphasize that all PML quantities depend on the prior $P_X$, whereas the ADP curve depends only on $k$ and $\varepsilon_r$.

\section{Other Related Works}
\label{sec:sota}
While the PML envelope has an information-theoretic flavor, it is most closely related to works proposing DP definitions, where post-processing robustness is treated as an axiom. Since the introduction of ADP, several other DP notions have been proposed, primarily to provide more convenient machinery for composition. Notable examples include \Ren DP~\cite{mironovRenyiDifferentialPrivacy2017a}, concentrated DP~\cite{bunConcentratedDifferentialPrivacy2016a}, and Gaussian DP~\cite{dongGaussianDifferentialPrivacy2022}. It is worth emphasizing that while modern systems often rely on these frameworks for privacy accounting (i.e., computing composition bounds during the internal calculations of an algorithm), the resulting guarantees are typically converted back to $(\varepsilon,\delta)$-DP for reporting and comparison purposes.

The work closest in spirit to ours is that of \textcite{ballePrivacyAmplificationSubsampling2018,ballePrivacyProfilesAmplification2020a} introducing \emph{privacy profiles}. The privacy profile is a function that precisely captures the tradeoff between $\varepsilon$ and $\delta$, similar to the PML envelope. Nonetheless, as explained in Section~\ref{sec:background}, $\delta$ in the privacy profile corresponds to an additive slack parameter, whereas in our setting it captures the worst-case failure probability.

\section{Conclusions and Future Work}
\label{sec:conclusions}
We introduced the PML envelope as a meaningful and robust privacy guarantee that captures the tradeoff between information leakage and failure probability under arbitrary post-processing. The present work focused on settings in which both the secret and the observable outcomes take values in finite probability spaces. This assumption provides a simple and mathematically convenient context for introducing the PML envelope. 

The two classes of mechanisms studied here (PML-extremal and $k$-RR) illustrate the different roles of the upper bounds in Theorem~\ref{thm:upper_bound_envelope}. In finite-alphabet settings, the second term is often the tighter one, especially when $\delta$ is small, since the maximal-leakage bound includes an additive $\log(1/\delta)$ term. The tightness of the second term is also discussed in Corollary~\ref{cor:constant_pml}. Still, the maximal-leakage bound has two useful features. First, it remains applicable in settings where the secret does not have finite support, where uniform $\varepsilon$-PML bounds usually do not exist. Second, maximal leakage is always bounded by $\log |\mathcal{Y}|$, independently of the mechanism and the prior on the secret, and therefore yields the general bound
\begin{equation*}
    \varepsilon_c(\delta) \leq \log \frac{\abs{\cY}}{\delta}, \quad \delta \in (0,1).
\end{equation*}
This provides a simple and meaningful privacy guarantee that holds for \emph{any} secret and \emph{any} computation with a finite output alphabet, even when no randomness is involved. Such bounds can be useful, for example, in statistical settings that involve deterministic recoding or quantization of the input data. The systematic study of these consequences is left for future work.

Another direction concerns the behavior of $\varepsilon_c$ under composition. Establishing sequential or adaptive composition results for the PML envelope beyond the maximal-leakage bounds remains an important topic for future work.

% In these settings, due to~\eqref{eq:pml_bounds}, the joint distribution $P_{XY}$ always satisfies $\varepsilon$-PML for some finite $\varepsilon$. The full force of the PML envelope, however, emerges in infinite settings, where such uniform bounds may not exist. Consequently, an important direction for future work is to study the envelope in more general contexts, for instance when the secret has infinite support or when unbounded noise e.g., Gaussian noise is added to the data. 

% in particular, the Gaussian mechanism, which is arguably the most commonly used mechanism both in terms of theoretical developments and in practical implementations.


% \section*{Acknowledgment}



\appendices
%%%%%%%%%%%%%%%%%%%%%%
%%%%%%%%%%%%%%%%%%%%%%
% \section{Equivalence of \texorpdfstring{\eqref{eq:lcq_func} and \eqref{eq:var_lcq_func}}{Eq. \ref{eq:lcq_func} and Eq. \ref{eq:var_lcq_func}}}
\section{Equivalence of \eqref{eq:lcq_func} and \eqref{eq:var_lcq_func}}
\label{ssec:quantile_equivalence_proof}
\begin{lemma}
\label{lemma:lcq_formulations}
For all $\delta \in (0,1)$, it holds that 
\begin{equation*}
    \inf \{t \geq 0 : C_Z(t) \geq 1 - \delta \} = \min_{\substack{\mathcal{A} \subset \cZ \\ P_Z(\mathcal{A}) \geq 1 - \delta}} \; \max_{z \in \mathcal{A}} \, \ell(z). 
\end{equation*}
\end{lemma}

\begin{IEEEproof}
Let 
\begin{gather*}
    r_1 \coloneqq \min_{\substack{\mathcal{A} \subset \cZ \\ P_Z(\mathcal{A}) \geq 1 - \delta}} \; \max_{z \in \mathcal{A}} \, \ell(z), \\[0.5em]
    r_2 \coloneqq \inf \{t \geq 0 : P_Z \{z : \ell(z) \leq t \} \geq 1 - \delta \}. 
\end{gather*} 

Given $r \geq 0$, let $\cJ_r = \big\{z : \ell(z) \leq r \big\}$, and note that $\max_{z \in \cJ_r} \ell(z) = r$. By definition, if $r < r_2$, then $P_Z(\cJ_r) < 1-\delta$, therefore $P_Z(\cJ_r) \geq 1-\delta$ for all $r \geq r_2$. Therefore, we have
\begin{align*}
    r_1 &= \min_{\substack{\mathcal{A} \subset \cZ \\ P_Z(\mathcal{A}) \geq 1 - \delta}} \; \max_{z \in \mathcal{A}} \, \ell(z)\\
    &\leq \inf_{\cJ_r : r \geq r_2} \; \max_{z \in \cJ_r} \, \ell(z)\\
    &\leq \inf_{r \geq r_2} \; r = r_2. 
\end{align*}

Next, we argue that the strict inequality $r_1 < r_2$ would lead to a contradiction. To see this, suppose $r_1 < r_2$. This means that there exists a set $\cE \subset \cZ$ with $P_Z(\cE) \geq 1 - \delta$ such that $r_1 \leq r_3 = \max_{z \in \cE} \ell(z) < r_2$. On the other hand, note that 
\begin{equation*}
    \cE \subseteq \big\{z : \ell(z) \leq r_3 \big\}, 
\end{equation*} 
therefore, 
\begin{equation*}
     P_Z \big\{z : \ell(z) \leq r_3 \big\} \geq P_Z (\cE) \geq 1- \delta. 
\end{equation*}
Hence, 
\begin{equation*}
    r_2 = \inf \{t \geq 0 : P_Z \{z : \ell(z) \leq t \} \geq 1 - \delta \} \leq r_3, 
\end{equation*}
which is a contradiction. We conclude that $r_1 = r_2$. 
\end{IEEEproof}

%%%%%%%%%%%%%%%%%%%%%%
%%%%%%%%%%%%%%%%%%%%%%
\section{Proof of Theorem~\ref{thm:comp_eps}}
\label{sec:thm_comp_eps_proof}
It is well-known that the right-continuous quantile function upper bounds the left-continuous quantile function. Nevertheless, we include a proof for completeness. Fix some random variable $Z$ with PML $\ell(Z)$, and $\delta \in (0,1)$. Let $0 \leq \alpha < \ubar \varepsilon_Z(\delta)$ and consider the set 
\begin{equation*}
    \cB_\alpha = \{z \in \cZ : \ell(z) \geq \alpha\}.
\end{equation*}
Then, $P_Z(\cB_\alpha) > \delta$,\footnote{This can be shown by contradiction: If $P_Z(\cB_\alpha) \leq \delta$, then the set $\cB_\alpha^\mathsf c = \cZ \setminus \cB_\alpha$ has probability 
$P_Z(\cB_\alpha^\mathsf c) \geq 1 - \delta$ and also satisfies $\max_{z \in \cB_\alpha^\mathsf c} \ell(z) \leq \alpha < \ubar \varepsilon_Z(\delta)$. This contradicts the definition of $\ubar \varepsilon_Z(\delta)$ as the smallest upper bound on the PML of all sets with probability at least $1-\delta$.} and we get 
\begin{equation*}
   \bar \varepsilon_Z(\delta) = \max_{\substack{\cA \subset \cZ: \\ P_Z(\cA) \geq \delta}} \quad \min_{z \in \cA} \; \ell(z) \geq \min_{z \in \cB_\alpha} \; \ell(z) \geq \alpha.  
\end{equation*}
Letting $\alpha \to \ubar \varepsilon_Z(\delta)$ yields
\begin{equation*}
   \bar \varepsilon_Z(\delta) \geq \ubar \varepsilon_Z(\delta). 
\end{equation*}

Next, we prove the opposite inequality, i.e.,  
\begin{equation*}
    \sup_{Z : X - Y - Z} \; \ubar \varepsilon_Z(\delta)
    \geq \sup_{Z : X - Y - Z} \; \bar \varepsilon_Z(\delta).
\end{equation*}
Fix $\delta \in (0,1)$ and suppose 
\begin{equation*}
    \ubar c = \ubar \varepsilon_Y(\delta) < \bar \varepsilon_Y(\delta) = \bar c, 
\end{equation*}
where $\bar c, \ubar c >0$. 
This happens if and only if there exists a subset $\cB \subset \cY$ with probability $P_Y(\cB) \geq \delta$ and $\min_{y \in \cB} \ell(y) = \bar c$ and $\cG = \cY \setminus \cB$ and $\max_{y \in \cG} \ell(y) = \ubar c$. Therefore, there exist outcomes $y_1 \in \cB, y_2 \in \cG$ such that 
\begin{gather}
\begin{split}
\label{eq:good_outcome_bad_outcome}
    \ell(X \to y_1) \geq \bar c, \\
    \ell(X \to y_2) \leq \ubar c. 
\end{split}
\end{gather}
Fix a parameter $\eta \in (0,1)$ and let $B \sim \mathrm{Bernoulli}(\eta)$ be independent of $(X,Y)$. Define $Z = h_\eta(Y,B)$ where 
\begin{equation*}
 h_{\eta}(Y,B) \;=\;
\begin{cases}
\bot,      & \text{if } Y = y_1,\\
\bot,      & \text{if } Y = y_2 \;\; \text{and} \;\; B = 1,\\
\diamond,  & \text{if } Y = y_2 \;\; \text{and} \;\; B = 0,\\
Y,         & \text{otherwise}.
\end{cases}
\end{equation*}
Note that $Z$ is a randomized function of $Y$ since it also depends on $B$. Now, observe that 
\begin{gather*}
    P_Z(\bot) = P_Y(y_1) + \eta P_Y(y_2), \\
    P_{Z \mid X=x}(\bot) =  P_{Y \mid X=x}(y_1) + \eta P_{Y \mid X=x}(y_2). 
\end{gather*}
Then, for fixed $x$ we have 
\begin{align*}
    &\frac{P_{Z \mid X=x} (\bot)}{P_Z(\bot)}= \frac{P_{Y \mid X=x}(y_1) + \eta P_{Y \mid X=x}(y_2)}{P_Y(y_1) + \eta P_Y(y_2)}\\[0.5em]
    &= \frac{\frac{P_{Y \mid X=x}(y_1)}{P_Y(y_1)} + \eta \left(\frac{P_{Y \mid X=x}(y_2)}{P_Y(y_2)}\right) \left(\frac{P_Y(y_2)}{P_Y(y_1)} \right)}{1 + \eta \frac{P_Y(y_2)}{P_Y(y_1)}}\\[0.5em]
    &= \Bigg( \frac{P_{Y \mid X=x}(y_1)}{P_Y(y_1)} + \eta \left(\frac{P_{Y \mid X=x}(y_2)}{P_Y(y_2)}\right) \cdot \left(\frac{P_Y(y_2)}{P_Y(y_1)} \right) \Bigg) \cdot\\
    &\hspace{12em} \left(1 - \eta \left(\frac{P_Y(y_2)}{P_Y(y_1)} \right) + O(\eta^2) \right)\\
    &= \frac{P_{Y \mid X=x}(y_1)}{P_Y(y_1)} - \eta \left(\frac{P_Y(y_2)}{P_Y(y_1)} \right) \cdot\\
    &\hspace{7em}\left( \frac{P_{Y \mid X=x}(y_1)}{P_Y(y_1)} -  \frac{P_{Y \mid X=x}(y_2)}{P_Y(y_2)} \right) + O(\eta^2). 
\end{align*}
Taking the logarithm and maximum over $x \in \cX$ on both sides gives 
\begin{align}
    &\ell_{P_{XZ}} (X \to \bot) = \log \; \max_x \frac{P_{Z \mid X=x} (\bot)}{P_Z(\bot)}\nonumber\\[0.5em]
    &= \log \; \max_x \Bigg(\frac{P_{Y \mid X=x}(y_1)}{P_Y(y_1)} - \eta \left(\frac{P_Y(y_2)}{P_Y(y_1)} \right) \cdot \nonumber
    \\
    &\hspace{6.5em}\left( \frac{P_{Y \mid X=x}(y_1)}{P_Y(y_1)} -  \frac{P_{Y \mid X=x}(y_2)}{P_Y(y_2)} \right) + O(\eta^2)\Bigg) \nonumber\\[0.5em]
    &\geq \log \Bigg( \max_x \frac{P_{Y \mid X=x}(y_1)}{P_Y(y_1)} - \eta \left(\frac{P_Y(y_2)}{P_Y(y_1)} \right) \cdot\nonumber\\
    &\hspace{4em}\max_{x'} \left( \frac{P_{Y \mid X=x'}(y_1)}{P_Y(y_1)} -  \frac{P_{Y \mid X=x'}(y_2)}{P_Y(y_2)} \right) + O(\eta^2) \Bigg),\nonumber\\[0.5em]
    &\geq \log \Big( e^{\bar c} - \eta \beta + O(\eta^2) \Big)\nonumber\\[0.5em]
    &= \bar c + \log \big(1-\eta\beta e^{-\bar c}+ O(\eta^2) \big), \label{subeq:lower_bound_symbol}
\end{align}
where
\begin{equation*}
    \beta = \left(\frac{P_Y(y_2)}{P_Y(y_1)} \right)\max_{x'} \left( \frac{P_{Y \mid X=x'}(y_1)}{P_Y(y_1)} -  \frac{P_{Y \mid X=x'}(y_2)}{P_Y(y_2)} \right).
\end{equation*}
Note that $\beta > 0$ because 
\begin{align*}
    \max_{x'} &\left( \frac{P_{Y \mid X=x'}(y_1)}{P_Y(y_1)} -  \frac{P_{Y \mid X=x'}(y_2)}{P_Y(y_2)} \right)\\
    &\geq \max_{x'}  \frac{P_{Y \mid X=x'}(y_1)}{P_Y(y_1)} -  \max_{x} \frac{P_{Y \mid X=x}(y_2)}{P_Y(y_2)}\\
    &\geq \exp({\bar c})  - \exp(\ubar c) > 0. 
\end{align*}
Now, using the elementary bound $\log(1-t) \geq -t-t^2$ for $0 < t < \tfrac12$ in~\eqref{subeq:lower_bound_symbol} yields
\begin{equation*}
    \ell_{P_{XZ}}(X \to \bot) \geq \bar c-\eta\beta e^{-\bar c}+O(\eta^2) = \bar c-\eta \gamma +O(\eta^2), 
\end{equation*}
with $\gamma = \beta e^{-\bar c}$. Thus, by taking $\eta \to 0$, we can bring $ \ell_{P_{XZ}}(X \to \bot)$ arbitrarily close to $\bar c$.

The final step is to argue that $ \ubar \varepsilon_Z(\delta) \geq \bar c$. To show this, let $\cB' = \cB \setminus \{y_1\}$ and $\cG' = \cG \setminus \{y_2\}$ so that the alphabet of $Z$ can be represented by $\cZ = \cB' \cup \cG' \cup \{\bot, \diamond \}$. Let $\cA$ be an arbitrary subset of $\cZ$ with probability $P_Z (\cA) \geq 1-\delta$. Since $P_Z(\cG' \cup \{\diamond\}) = P_Y(\cG) - \eta P_Y(y_2) < 1- \delta$, any such set $\cA$ either intersects with $\cB'$ or contains $\bot$. If $\cA$ contains elements from $\cB'$, then 
\begin{equation*}
     \max_{z \in \cA} \ell_{P_{XZ}}(X \to z) \geq \min_{z \in \cB'} \ell_{P_{XZ}}(X \to z) \geq \bar c, 
\end{equation*}
and if $\bot \in \cA$, then 
\begin{equation*}
    \max_{z \in \cA} \ell_{P_{XZ}}(X \to z) \geq \ell_{P_{XZ}} (X \to \bot) \geq \bar c-\eta \gamma +O(\eta^2), 
\end{equation*}
and taking $\eta \to 0$ yields $\max_{z \in \cA} \ell_{P_{XZ}}(X \to z) \geq \bar c$. Hence, we have proved that 
\begin{align*}
    \ubar \varepsilon_Z(\delta) &= \min_{\mathcal{A} :P_Z(\mathcal{A}) \geq 1 - \delta} \; \max_{z \in \mathcal{A}} \, \ell(X \to z) \geq \bar c =  \bar \varepsilon_Y(\delta).
\end{align*}
This, in turn, implies that 
\begin{equation*}
    \sup_{Z : X - Y - Z} \ubar \varepsilon_Z(\delta) \geq \sup_{Z : X - Y - Z} \quad \bar \varepsilon_Z(\delta),
\end{equation*}
as desired.

%%%%%%%%%%%%%%%%%%%%%%
%%%%%%%%%%%%%%%%%%%%%%
\section{Measuring the Information Leaked to Events}
\label{sec:events}
Given an event $\cE \subseteq \cY$ with $P_Y(\cE) >0$ and a random variable $Z = \ind_{\mathcal{E}}(Y)$, we define 
\begin{equation*}
    \ell_{P_{XY}}(X \to \mathcal{E}) \coloneqq \ell_{P_{XZ}}(X \to 1).  
\end{equation*}
This definition is motivated by the observation that both deterministic and randomized post-processings can be naturally expressed in terms of indicator functions. For deterministic mappings, this is immediate: If $Z = h(Y)$ and $z \in \cZ$, then 
\begin{equation*}
    \ell_{P_{XZ}}(X \to z) = \ell_{P_{XY}}(X \to \cE_z),
\end{equation*}
where $\cE_z = \{y \in \cY : h(y) = z\}$ denotes the pre-image of $z$ under $h$.

We can also cast randomized mappings as deterministic ones by adopting a few formalisms from \cite{mciverAbstractChannelsTheir2014} and \cite{saeidian2023pointwise_it}. Given a privacy mechanism $P_{Y \mid X}$, we say that two outcomes $y, y'$ are \emph{similar} if there exists a constant $c>0$ such that $P_{Y \mid X=x}(y) = c P_{Y \mid X=x}(y')$ for all $x \in \cX$. Similar outcomes have the same information density $i(x;y) = i(x;y')$ for all $x \in \cX$, and induce the same posterior distributions $P_{X \mid Y=y} = P_{X \mid Y=y'}$. Consequently, \say{merging} similar outcomes (i.e., mapping similar outcomes to the same symbol) does not alter the distribution of $\ell(X \to Y)$. In \cite{mciverAbstractChannelsTheir2014}, the mechanism obtained by merging all similar outcomes is called the \emph{reduced mechanism}. As an example, for the mechanism $P_{Y \mid X}$ in~\eqref{eq:ex_mech}, outcomes $3$ and $4$ are similar, and its reduced form is 
\begin{equation*}
    P_{Y_r \mid X} = \begin{bmatrix}
    0 & 0 & 1\\[0.5em]
    0 & 0 & 1\\[0.5em]
    0 & 0.2 & 0.8\\[0.5em]
    0.2 & 0 & 0.8
    \end{bmatrix}. 
\end{equation*}

Next, we define an equivalence relation that unifies all mechanisms with the same reduced form. Let $[P_{Y\mid X}]$ denote the equivalence class of $P_{Y\mid X}$. A key advantage of introducing such equivalence classes is that a randomized post-processing applied to the outputs of $P_{Y\mid X}$ can alternatively be viewed as a deterministic post-processing applied to some mechanism in $[P_{Y\mid X}]$. While this idea can be established formally, we illustrate it with a simple example that readily extends to a general proof. Fix a mechanism $P_{Y\mid X}$ with binary output alphabet $\mathcal Y=\{0,1\}$, and consider a randomized post-processing $P_{Z\mid Y}$ of the form
\begin{equation*}
    P_{Z \mid Y} = \begin{bmatrix}
    \alpha & 1 - \alpha\\[0.5em]
    1 - \beta & \beta
    \end{bmatrix}, 
\end{equation*}
with $0<\alpha, \beta < 1$ and $\cZ = \{0,1\}$. Define a mechanism $P_{\tilde Y \mid X}$ with the output space $\tilde \cY = \{00, 01, 10, 11\}$ by 
\begin{gather*}
    P_{\tilde Y \mid X=x}(00) = \alpha P_{Y \mid X=x}(0),\\ P_{\tilde Y \mid X=x}(01) = (1 - \alpha) P_{Y \mid X=x}(0),\\
    P_{\tilde Y \mid X=x}(10) = (1-\beta) P_{Y \mid X=x}(1),\\ P_{\tilde Y \mid X=x}(11) = \beta P_{Y \mid X=x}(1),
\end{gather*}
for all $x$. Then, $P_{\tilde Y \mid X} \in [P_{Y \mid X}]$ and we have
\begin{gather*}
\begin{split}
    \ell_{P_{XZ}}(X \to 0) =  \ell_{P_{X \tilde Y}}(X \to \{00, 10\}),\\
    \ell_{P_{XZ}}(X \to 1) =  \ell_{P_{X \tilde Y}}(X \to \{01, 11\}), 
\end{split}
\end{gather*}
that is, $Z$ is a deterministic function of $\tilde Y$. Extending this construction to general post-processings shows that any randomized post-processing of $Y$ can be treated as a deterministic mapping applied to some mechanism in the equivalence class $[P_{Y\mid X}]$. 

Moreover, using equivalent classes, we interpret ``events'' in the generalized sense of applying an indicator function to some $\tilde Y$ induced by a mechanism $P_{\tilde Y \mid X} \in [P_{Y \mid X}]$. Concretely, we extend the notation by defining
\begin{equation}
\label{eq:generalized_event}
    \ell_{P_{XY}}(X\to \mathcal E) \coloneqq \ell_{P_{XZ}}(X\to 1),
\end{equation}
where $Z=\ind_{\mathcal E}(\tilde Y)$ for $\cE \subseteq \tilde \cY$ and $P_{\tilde Y\mid X}\in [P_{Y\mid X}]$.

Below, we establish some elementary properties of the map $\cE \mapsto \ell_{P_{XY}}(X\to \mathcal E)$. As a preliminary step, let us express $\ell_{P_{XY}}(X\to \mathcal E)$ in terms of the information density:
\begin{align}
    &\ell_{P_{XY}}(X \to \mathcal{E}) = \log \; \max_{x \in \cX} \; \frac{P_{Y \mid X=x}(\cE)}{P_Y(\cE)} \nonumber\\
    &= \log \; \max_{x \in \cX} \; \frac{\sum_{y \in \cE} P_{Y \mid X=x}(y)}{P_Y(\cE)} \nonumber\\
    &= \log \; \max_{x \in \cX} \; \frac{\sum_{y \in \cE} \exp \big(i(x;y)\big) \, P_Y(y)}{P_Y(\cE)} \nonumber\\
    % &= \log \; \max_{x \in \cX} \; \sum_{y \in \cE} \exp \big(i(x;y)\big) \, \frac{P_Y(y)}{P_Y(\cE)} \nonumber\\
    &= \log \; \max_{x \in \cX} \; \bE_{Y \sim Q_{\cE}} \Big[\exp \big(i(x;Y)\big)\Big], \label{eq:eml_as_expectation}
\end{align}
where $Q_{\cE}$ is the conditional distribution of $Y$ given $\cE$, that is, 
\begin{equation*}
    Q_{\cE}(y) = 
    \begin{cases}
        \frac{P_Y(y)}{P_Y(\cE)} & \text{ if } y \in \cE,\\[0.5em]
        0 & \text { otherwise.} 
    \end{cases}    
\end{equation*}

\begin{lemma}
\label{lemma:properties_event_leakage}
The function $\ell(X \to \cE)$ satisfies the following properties: 
\begin{enumerate}
    \item $0 \leq \ell(X \to \cE) \leq \log \frac{1}{P_Y(\cE)}$ for all $\cE \subseteq \cY$. 
    \item $\ell(X \to \cY) = 0$. 
    \item If $\cE \cap \cE' = \emptyset$, then
    \begin{equation*}
        \ell(X \to \cE \cup \cE') \leq \max \Big\{ \ell(X \to \cE), \ell(X \to \cE') \Big\}.
    \end{equation*}
    \item Suppose $P_Y(\cE) = \theta >0$. For each $0 < \theta' < \theta$, there exists an event $\cE' \subset \cE$ with probability $P_Y(\cE') = \theta'$ such that 
\begin{equation*}
    \ell(X \to \cE') \geq \ell(X \to \cE). 
\end{equation*}
\end{enumerate}    
\end{lemma}

\begin{IEEEproof}
\begin{enumerate}
\item These bounds follow immediately from the definition of $\ell(X \to \cE)$ and were also noted in \cite{saeidian2023pointwise_it}.
\item 
\begin{align*}
    \ell_{P_{XY}}(X\to \mathcal Y) = \log \max_{x\in\mathcal X}\frac{P_{Y\mid X=x}(\mathcal Y)}{P_Y(\mathcal Y)} = \log\frac{1}{1} = 0.
\end{align*}
\item Suppose $\cE \cap \cE' = \emptyset$. We have 
\begin{align*}
    &\ell(X \to \cE \cup \cE') = \log \, \max_x \frac{P_{Y \mid X=x}(\cE \cup \cE')}{P_{Y}(\cE \cup \cE')}\\
    &= \log \, \max_x \frac{P_{Y \mid X=x}(\cE) + P_{Y \mid X=x}(\cE') }{P_{Y}(\cE) + P_{Y}(\cE')}\\
    &\leq \log \, \max_x \, \max \Big\{\frac{P_{Y \mid X=x}(\cE)  }{P_{Y}(\cE)}, \frac{P_{Y \mid X=x}(\cE')  }{P_{Y}(\cE')} \Big\}\\
    &= \max \Big\{ \ell(X \to \cE), \ell(X \to \cE') \Big\}.
\end{align*}

%%%%%%%%%%%%%%%%
\item 
Fix $x \in \cX$ satisfying 
\begin{equation*}
    \log \frac{P_{Y \mid X=x}(\cE)}{P_Y(\cE)} = \ell(X \to \cE). 
\end{equation*}
Choose $\tau >0$ so that 
\begin{align*}
    P_Y \{y \in \cE : \exp &\big(i(x;y)\big) > \tau \} \leq \theta'\\
    &\leq P_Y \{y \in \cE : \exp \big(i(x;y)\big) \geq \tau \}. 
\end{align*}
Let $\cA = \{y \in \cE : \exp \big(i(x;y)\big) > \tau \}$. If $P_Y(\cA) < \theta'$, choose a set $\cB \subseteq \{y \in \cE : \exp \big(i(x;y)\big) = \tau \}$ so that $P_Y(\cA) + P_Y(\cB) = \theta'$ and let $\cE' = \cA \cup \cB$.\footnote{To select such $\cB$, we might need to use some other mechanism in $[P_{Y \mid X}]$.} If $P_Y(\cA)  = \theta'$, then let $\cE' = \cA$. By construction, we have 
\begin{align*}
    \exp\big(i(x;y)\big) &\geq \tau, \quad y \in \cE', \\
    \exp\big(i(x;y)\big) &\leq \tau, \quad y \in \cE \setminus \cE'.  
\end{align*}

Now, we can write  
\begin{subequations}
\begin{align}
    &\ell(X \to \cE) = \log \,\frac{P_{Y \mid X=x}(\cE)}{P_Y(\cE)} \nonumber\\
    &= \log \, \frac{P_{Y \mid X=x}(\cE') + P_{Y \mid X=x}(\cE \setminus \cE')}{P_Y(\cE') + P_Y(\cE \setminus \cE')} \nonumber\\[0.5em]
    &\leq \log \, \max \Big\{\frac{P_{Y \mid X=x}(\cE')}{P_Y(\cE')}, \; \frac{P_{Y \mid X=x}(\cE \setminus \cE')}{P_Y(\cE \setminus \cE')} \Big\} \nonumber\\[0.5em]
    &= \log \, \frac{P_{Y \mid X=x}(\cE')}{P_Y(\cE')} \label{subeq:bigger_leakage_event} \\[0.5em] 
    &\leq \ell(X \to \cE'), \nonumber
\end{align}
\end{subequations}
where \eqref{subeq:bigger_leakage_event} follows because
\begin{align*}
    \frac{P_{Y \mid X=x}(\cE')}{P_Y(\cE')} &= \bE_{Y \sim Q_{\cE'}} \Big[\exp\big(i(x;Y)\big)\Big] \geq \tau,\\
    \frac{P_{Y \mid X=x}(\cE \setminus \cE')}{P_Y(\cE \setminus \cE')} &= \bE_{Y \sim Q_{\cE \setminus \cE'}} \Big[\exp\big(i(x;Y)\big)\Big] \leq \tau. 
\end{align*} 
\end{enumerate}
\end{IEEEproof}
We make some further remarks about the event-wise leakage. First, in general, the reverse of Lemma~\ref{lemma:properties_event_leakage}(iii) need not hold, i.e., one cannot claim that $\ell(X \to \cE \cup \cE')$ upper bounds $\min \Big\{ \ell(X \to \cE), \ell(X \to \cE') \Big\}$.

\begin{example}
\label{ex:event_properties}
Let $X$ be an unbiased Bernoulli random variable, and let $\mathcal Y=\{1,2,3\}$. Consider the mechanism
\begin{equation*}
    P_{Y\mid X}=
    \begin{bmatrix}
    0.9 & 0   & 0.1\\
    0   & 0.9 & 0.1
    \end{bmatrix},
\end{equation*}
which induces the marginal distribution $P_Y$ with 
\begin{equation*}
    P_Y(1)=0.45,\qquad P_Y(2)=0.45,\qquad P_Y(3)=0.1.
\end{equation*}
Let $\mathcal E=\{1\}$ and $\mathcal E'=\{2\}$. Then, we have 
\begin{align*}
    \ell(X\to \mathcal E) &= \ell(X\to \mathcal E')\\
    &=\log \max_{x\in\{0,1\}}\frac{P_{Y\mid X=x}(\{1\})}{P_Y(\{1\})}\\
    &=\log\frac{0.9}{0.45}\\
    &=\log 2.
\end{align*}
However, for their union $\mathcal E\cup \mathcal E'=\{1,2\}$ we have
\begin{equation*}
    \ell(X\to \mathcal E\cup \mathcal E') =\log \max_{x\in\{0,1\}}\frac{P_{Y\mid X=x}(\{1,2\})}{P_Y(\{1,2\})} =\log 1 =0.
\end{equation*}
Thus, $\ell(X\to \mathcal E\cup \mathcal E')=0<\min\{\ell(X\to \mathcal E),\,\ell(X\to \mathcal E')\}=\log 2$.
\end{example}

Second, the event-wise leakage is, in general, not monotone. That is, given events $\cE \subset \cE'$, either $\ell(X \to \cE)$ or $\ell(X \to \cE')$ can be larger. 
\begin{example}
Recall the setup of Example~\ref{ex:event_properties}. Let $\mathcal E=\{3\}$ and $\mathcal E'=\{1,3\}$, so that $\mathcal E\subset \mathcal E'$. Then, we have 
\begin{equation*}
    \ell(X\to \mathcal E) = \log\max_{x}\frac{P_{Y\mid X=x}(\{3\})}{P_Y(\{3\})} = \log\frac{0.1}{0.1} =0,
\end{equation*}
whereas
\begin{equation*}
    \ell(X\to \mathcal E')
    = \log\max_{x}\frac{P_{Y\mid X=x}(\{1,3\})}{P_Y(\{1,3\})}
    = \log\frac{1}{0.55}
    >0.
\end{equation*}

Now, let $\mathcal F=\{1\}$ and $\mathcal F'=\{1,3\}$, so that $\mathcal F\subset \mathcal F'$. We have
\begin{equation*}
    \ell(X\to \mathcal F) = \log 2,
\end{equation*}
while, as computed above,
\begin{equation*}
    \ell(X\to \mathcal F')=\log \left(\frac{1}{0.55}\right)<\log 2.
\end{equation*}
\end{example}

\subsection{Connection to Prior Work and Derivation of  Algorithm~\ref{alg:epsb}}
The binary envelope and the procedure in Algorithm~\ref{alg:epsb} are closely related to a privacy guarantee studied in \cite{saeidian2023pointwise_it}. In particular, \textcite{saeidian2023pointwise_it} introduced a guarantee based on the worst-case PML of post-processed outcomes with probability at least $\delta$, namely the quantity
\begin{equation}
\label{eq:old_work_eml}
    \sup_{Z:X-Y-Z} \; \max_{z\in\mathcal Z:P_Z(z)\geq \delta} \ell_{P_{XZ}}(X\to z).
\end{equation}
Then, they showed that \eqref{eq:old_work_eml} admits the equivalent formulation
\begin{equation}
\label{eq:eml}
    \sup_{\cE \subseteq \cY : P_Y(\cE) \geq \delta} \ell_{P_{XY}}(X \to \cE),  
\end{equation}
where the supremum is over all events in the generalized sense described above~\cite[Thm.~3]{saeidian2023pointwise_it}. It follows immediately from \eqref{eq:generalized_event} that both~\eqref{eq:old_work_eml} and \eqref{eq:eml} coincide with the binary PML envelope $\varepsilon_b$ defined in~\eqref{eq:binary_envelope}.

% It is straighforward to see that $\eqref{eq:old_work_eml}$ and $\varepsilon_b$ in \eqref{eq:binary_envelope}, observe that although \eqref{eq:old_work_eml} optimizes over arbitrary post-processings $P_{Z\mid Y}$ (possibly with non-binary output alphabet) the optimization in fact reduces to binary post-processings. Indeed, fix any $Z$ and let $z^\star$ attain the inner maximum in \eqref{eq:old_work_eml}. Define a binary post-processing $B$ by
% \begin{equation*}
%     P_{B\mid Y}(1\mid y)\coloneqq P_{Z\mid Y}(z^\star\mid y), \quad y\in\mathcal Y.
% \end{equation*}
% Then, $P_B(1)=P_Z(z^\star)\ge\delta$ and $P_{B\mid X=x}(1)=P_{Z\mid X=x}(z^\star)$ for all $x$, which implies
% \begin{equation*}
%     \ell_{P_{XB}}(X\to 1)=\ell_{P_{XZ}}(X\to z^\star).
% \end{equation*}
% Consequently, restricting the mechanism in \eqref{eq:old_work_eml} to those with binary outcomes is  without loss of optimality in. It then becomes evident that \eqref{eq:old_work_eml} coincides with the binary envelope $\varepsilon_b(\delta)$. 

\textcite{saeidian2023pointwise_it} also characterized the solution to \eqref{eq:old_work_eml} and \eqref{eq:eml}. Fix $x\in\mathcal X$ and order the outputs $y\in\mathcal Y$ in decreasing information density $i(x;y)$. Let $\mathcal F_k$ denote the set consisting of the first $k$ outputs in this ordering, and let $k^\star$ be the smallest index such that $P_Y(\mathcal F_{k^\star})\geq \delta$. If $P_Y(\mathcal F_{k^\star})>\delta$, the optimal construction uses randomization at the boundary output so that the selected event has probability exactly $\delta$. For each $x \in \cX$, this yields the value
\begin{equation}
\label{eq:eml_opt_val}
    \kappa(x) =\frac{1}{\delta}\Bigl(P_{Y\mid X=x}(\mathcal F_{k^\star-1}) +\zeta\,P_{Y\mid X=x}(y_{k^\star})\Bigr),
\end{equation}
where $\zeta\in(0,1]$ is chosen so that
$P_Y(\mathcal F_{k^\star-1})+\zeta P_Y(y_{k^\star})=\delta$.
Finally, the binary envelope is obtained as
\begin{equation*}
    \varepsilon_b(\delta)=\log \, \max_{x\in\mathcal X}\, \kappa(x).
\end{equation*}
Algorithm~\ref{alg:epsb} formalizes this procedure.

%%=====================%
%%=====================%
\subsection{Binary Envelope of Randomized Response}
\label{ssec:binary_envelope_krr}
Recall that $P_Y(y) = q_y$ for $y \in [k]$, 
\begin{equation*}
    q_1 \leq q_2 \leq \dots \leq q_k, 
\end{equation*}
and 
\begin{equation*}
    \sum_{y=1}^{N-1} q_y < \delta \le \sum_{y=1}^{N} q_y,
    \qquad
    \sum_{y=1}^{N-1} q_y + \theta q_N = \delta,
\end{equation*}
where $\delta \in (q_1, 1)$ and $N \geq 2$. To compute $\varepsilon_b$ for the $k$-RR mechanism, we use the equivalent formulation \eqref{eq:eml}, which leads to a simpler argument than applying Algorithm~\ref{alg:epsb}. Our goal is to compute the largest value of
\begin{equation*}
    \ell_{P_{XY}}(X \to \cE)
    = \log \max_{x \in [k]} \frac{P_{Y \mid X=x}(\cE)}{P_Y(\cE)},
\end{equation*}
over all (possibly randomized) events $\cE$ satisfying $P_Y(\cE) \geq \delta$. Note that by Lemma~\ref{lemma:properties_event_leakage}(iv), we can restrict attention to events with probability exactly $\delta$. 

Fix $x \in [k]$ and write
\begin{gather*}
    P_{Y \mid X=x}(\cE) = \sum_{y \in [k]} w_y P_{Y \mid X=x}(y),\\
    P_Y(\cE) = \sum_{y \in [k]} w_y P_Y(y),
\end{gather*}
where $w_y \in [0,1]$ denotes the probability that outcome $y$ is included in the event $\cE$. This yields the optimization problem
\begin{align*}
    \max_{x \in [k]} \ \max_{w \in [0,1]^k} \ \sum_{y \in [k]} w_y P_{Y \mid X=x}(y),\\
    \text{subject to } \sum_{y \in [k]} w_y q_y = \delta .
\end{align*}

We next use the following claim to upper bound the objective; its proof is deferred to the end of this section.

\smallskip
\noindent\underline{Claim:}
For all $w \in [0,1]^k$ satisfying $\sum_y w_y P_Y(y) = \delta$, we have
\begin{equation*}
    \sum_y w_y \le N-1 + \theta .
\end{equation*}

\smallskip
Recall that $P_{Y \mid X=x}(y) = \alpha$ if $y = x$ and $P_{Y \mid X=x}(y) = \beta$ otherwise. Fixing $x \in [k]$ and applying the claim, we obtain
\begin{align*}
    \sum_{y \in [k]} w_y P_{Y \mid X=x}(y)
    &= \alpha w_x + \beta \sum_{y \neq x} w_y\\
    &= (\alpha - \beta) w_x + \beta \sum_{y \in [k]} w_y\\
    &\le (\alpha - \beta) w_x + \beta (N-1 + \theta)\\
    &\le \alpha - \beta + \beta (N-1 + \theta)\\
    &= \alpha + \beta (N-2 + \theta).
\end{align*}

Observe that this upper bound is achieved by taking $x = 1$ and choosing
\begin{equation*}
    w_1 = \cdots = w_{N-1} = 1, 
    \qquad 
    w_N = \theta,
\end{equation*}
which is feasible by construction. We therefore conclude that
\begin{align*}
    \sup_{\cE : P_Y(\cE) \ge \delta} \, \ell_{P_{XY}}(X \to \cE)
    &= \log \max_{x \in [k]} \frac{P_{Y \mid X=x}(\cE)}{P_Y(\cE)}\\
    &= \log \frac{\alpha + \beta (N-2 + \theta)}{\delta}.
\end{align*}

The final piece is to prove the claim. Consider the optimization problem
\begin{gather*}
    \max_{w \in [0,1]^k} \ \sum_{y=1}^k w_y \\
    \text{subject to } \ \sum_{y=1}^k q_y w_y = \delta .
\end{gather*}
This is a linear program, hence it has an optimal solution at an extreme point of the feasible polytope. In particular, at an extreme point all but at most one coordinate of $w$ is in $\{0,1\}$.

Since each coordinate $w_y$ contributes equally to the objective, maximizing $\sum_y w_y$ amounts to setting as many $w_y$'s to 1 as possible. Because $q_1 \le q_2 \le \cdots \le q_k$, the optimal strategy is to set $w_1=1, w_2=1, \dots$ until the constraint $\sum_y w_y q_y = \delta$ is met. Thus, the choice
\begin{equation*}
    w_1 = \cdots = w_{N-1} = 1,\quad w_N=\theta, \quad w_{N+1}= \dots = w_k = 0,
\end{equation*}
is both feasible and optimal. Therefore, any feasible $w$ satisfies
\begin{equation*}
    \sum_{y=1}^k w_y \leq N-1+\theta,
\end{equation*}
as desired.

% In the third statement of Lemma~\ref{lemma:properties_event_leakage}, the parameter $\eta >0$ is required in case the supremum over $x$ in the expression of $\ell (X \to \cE)$ is not attained. If the supremum is attained, then we can take $\eta = 0$. 

\newpage
% \clearpage
\printbibliography

\end{document}